%% file: daggeometrypaths.tex
\documentclass[12pt,a4paper]{article}   

\usepackage[utf8]{inputenc}

\usepackage{subcaption}

\usepackage[UKenglish]{isodate}

\usepackage[a4paper,margin=2cm, marginparwidth=2cm]{geometry}

\usepackage{amsmath,amssymb,amscd}
\numberwithin{equation}{section}

\usepackage{amsthm}
\newtheorem{theorem}{Theorem}[section]

\newtheorem{lemma}[theorem]{Lemma}

\newtheorem{corollary}[theorem]{Corollary}

\newtheorem{definition}{Definition}[section]

\usepackage{graphicx}

\graphicspath{{./figures/}}

\usepackage[dvipsnames]{xcolor} 

\usepackage{float}

\usepackage[round,sort,comma,authoryear]{natbib}

\renewcommand{\harvardurl}[1]{\textbf{URL:} \url{#1}}
\providecommand{\citep}[1]{\cite{#1}}
\providecommand{\citet}[1]{\cite{#1}}

\providecommand{\eqref}[1]{(\ref{#1})}

\newcommand{\figref}[1]{Fig.\ \ref{#1}}
\newcommand{\Figref}[1]{Fig.\ \ref{#1}} 
\newcommand{\tabref}[1]{Table~\ref{#1}}
\newcommand{\secref}[1]{Section~\ref{#1}}
\newcommand{\Secref}[1]{Section~\ref{#1}} 
\newcommand{\appref}[1]{Appendix~\ref{#1}}

\newcommand{\lemref}[1]{Lemma~\ref{#1}}
\newcommand{\defref}[1]{Definition~\ref{#1}}

\usepackage[colorlinks,allcolors=blue]{hyperref}

\providecommand{\url}[1]{\texttt{#1}}

\newcommand{\tsenote}[1]{}
\newcommand{\tsecomment}[1]{}
\newcommand{\tsetodo}[1]{}

\renewcommand{\tsenote}[1]{\textbf{\color{red} (T)}\footnote{\textbf{\color{red} (T)} \color{BrickRed}{#1}}}
\renewcommand{\tsecomment}[1]{\textbf{\color{red} (Start of T)} { \color{BrickRed} #1} {\color{red} (end of T)} }
\renewcommand{\tsetodo}[1]{\textbf{\color{red} (Start of TODO)} { \color{BrickRed} #1} {\color{red} (end of TODO)} }

\newcommand{\tsedef}[1]{\textsc{#1}}

\newcommand{\half}{\frac{1}{2}}
\newcommand{\bea}{\begin{eqnarray}}
\newcommand{\eea}{\end{eqnarray}}
\newcommand{\beq}{\begin{equation}}
\newcommand{\eeq}{\end{equation}}

\newcommand{\tsevec}[1]{\mathbf{#1}}

\newcommand{\xvec}{\tsevec{x}}
\newcommand{\yvec}{\tsevec{y}}

\newcommand{\Rbb}{\mathbb{R}}

\newcommand{\Ccal}{\mathcal{C}}
\newcommand{\Dcal}{\mathcal{D}}
\newcommand{\Ecal}{\mathcal{E}}
\newcommand{\Gcal}{\mathcal{G}}

\newcommand{\Pcal}{\mathcal{P}}
\newcommand{\Rcal}{\mathcal{R}}

\newcommand{\Vcal}{\mathcal{V}}

\newcommand{\Xcal}{\mathcal{X}}

\newcommand{\dedge}[2]{(#1 \rightarrow #2)} 
\newcommand{\kout}{k^\mathrm{(out)}}
\newcommand{\Ntilde}{\tilde{N}}
\newcommand{\Nbox}{N}
\newcommand{\Ntotal}{\tilde{N}}

\newcommand{\Lhop}{L^\mathrm{(hop)}}
\newcommand{\Lweight}{L^\mathrm{(W)}}

\newcommand{\tseexpect}[1]{\langle #1 \rangle}

\newcommand{\gammamin}{\gamma^\mathrm{(min)}}
\newcommand{\gammamax}{\gamma^\mathrm{(max)}}

\newcommand{\Lp}{L\textsubscript{p}}
\newcommand{\geolength}{L_p(\gamma_\mathrm{O})}

\newcommand{\pointset}{\Xcal}
\newcommand{\nodeset}{\Vcal}
\newcommand{\pathset}{\Pcal}
\newcommand{\walkseq}{W}

\DeclareMathOperator*{\argmax}{arg\,max}
\DeclareMathOperator*{\argmin}{arg\,min}

\begin{document}

\renewcommand{\thefootnote}{\fnsymbol{footnote}}

\begin{center}
 {\Large\textbf{A Geometric Principle for Path Selection}}
 \\
 {\Large\textbf{in Directed Acyclic Graphs}}
 \\[\baselineskip]
 H.\ AlMuhanna\textsuperscript{1,2}\footnote{ORCID:  \href{http://orcid.org/0009-0004-2564-0140}{\texttt{0009-0004-2564-0140}} }, 
 \href{http://www.imperial.ac.uk/people/t.evans}{T.S.\ Evans}\textsuperscript{1,2}\footnote{ORCID:  \href{http://orcid.org/0000-0003-3501-6486}{\texttt{0000-0003-3501-6486}} }
 \\[0.5\baselineskip]
 1.\ \href{http://complexity.org.uk/}{Centre for Complexity Science}, Imperial College London, SW7 2AZ, U.K.
 \\
 2.\ \href{http://www3.imperial.ac.uk/theoreticalphysics}{Abdus Salam Centre for Theoretical Physics},  Imperial College London, SW7 2AZ, U.K.
 \\[0.5\baselineskip]
 9\textsuperscript{th} October 2026

\end{center}

\begin{abstract}
There are many types of path between nodes of a Directed acyclic graph (DAG), including the shortest and longest paths. Here, we ask whether the choice of path used in any analysis reflects the geometry of the latent space in which the DAG is embedded. To do this we define \Lp{}-cube spacetimes, a family of continuous ordered spaces in which a single parameter $p$ determines whether the ordered distance satisfies the ordered triangle identity when $p>1$, or the reverse ordered triangle identity when $0<p<1$ or $p<0$, both of which are rarely studied. The forward connection kernel in the $p<0$ case has infinite hyper-volume, yet it is a well defined model.
We then construct random geometric DAGs embedded in our \Lp{}-cube spacetime and study their geometry. 
For $p>1$, where the ordered triangle identity holds, the shortest paths provide the closest approximation to the continuum ordered geodesic while for $0<p<1$ and $p<0$, where the reverse ordered triangle identity holds, this role is instead taken by the longest weighted path. We also test the robustness of this relationship in a geometric citation model in which the random geometric DAG is supplemented by a preferential-attachment-type mechanism allowing connections beyond the strict geometric cut-off. The same broad correspondence between latent geometry and ordered geodesic path type is recovered. Our results show that the preference for shortest or longest paths in a DAG is closely connected to the geometry assumed for its latent space, and suggest that path structure can provide information about latent geometry, even when that geometry is not directly observed.
\end{abstract}

\setcounter{footnote}{0}
\renewcommand{\thefootnote}{\arabic{footnote}}

\section{Introduction} \label{s:introduction}

In directed acyclic graphs (DAGs), there are many types of paths between two nodes which can be used to analyse the network. In citation networks, for example, longest paths and related constructions such as main-path analysis identify extended innovation trajectories through the network \citep{HD89,HPEO24,PE25}, while the intermediacy measure \citep{SWTE20} favours shortest paths in one limit. 
However, there is no single principled reason to prefer one path notion over another in general. 

In this work, we take a different perspective: rather than asking which path is intrinsically ``most meaningful'', we ask whether the preference for shortest or longest paths reflects a choice of distinct types of geometry for the DAG.
The motivation for this question comes from the relationship between networks and latent geometry. In machine learning it is often useful to embed data in a continuous space equipped with a notion of distance, typically a metric space such as Euclidean or hyperbolic space \citep[for example][]{MSCCD13,BUGWY13,NK17,KPKVB10,HRH02}. These spaces, however, have no intrinsic sense of order: there is no sense of moving forward or backwards, no natural partial order on their points. DAGs are different. A DAG defines a partial order on its nodes, and in a continuous-space description this is analogous to introducing a time coordinate which distinguishes past from future. Temporal and spatial coordinates can therefore be treated together within a spacetime description. Embedding a DAG in such a continuous ordered space allows geometric methods to be used to study its structure. 
In our study we develop a new ordered space, the \Lp{}-cube spacetime, as a family of continuous spaces whose geometry is controlled by a single parameter. We then create DAGs embedded in this space. By understanding the relationship between geometry and network paths, we can ask what may be inferred when a DAG is observed without any explicit latent space.

There is an established literature on DAGs embedded in ordered continuous spaces. Some of this work is motivated by relativity and cosmology, particularly the causal-set approach to quantum gravity \citep{M78,BLMS87,BB91,BG91,D05,D06a,RW09,S19b}. Related spacetime constructions have also been applied to citation networks \citep{CGLE14,CE16,CE17} or developed for DAGs in general \citep{SKLDB25}. Minkowski and de Sitter spacetimes provide natural examples of continuous spaces with a distinction between temporal and spatial directions. They also come equipped with a natural measure of separation, proper time. However, these highly symmetric spacetimes are more restrictive than we require here. We therefore construct a different ordered space, the \Lp{}-cube spacetime, which allows us to vary the geometric properties of the ordered distance within a single framework. This makes it possible to study both ordinary and reverse triangle identities, and hence geometries in which either shortest or longest paths play the role of the continuum ordered geodesic.

The paper starts in \secref{s:Lpcsmodel} where we introduce the \Lp{}-cube spacetime and we define how we measure distances and define geodesics in such a space. In \secref{s:dagst} we define the DAGs embedded in this space and the paths we use to analyse these DAGs. \Secref{s:results} gives the results of the comparisons we make between the different paths in the DAGs and the geodesics of the \Lp{}-cube space model. We conclude with a discussion of the implications of these results and possible extensions in \Secref{s:conclusion}.

\section{\Lp{}-cube spacetime}\label{s:Lpcsmodel}

\subsection{Cube-space rule}\label{s:csrule}

We start from the work of \citet{BB91} on \tsedef{cube spacetime} which is the set of points given by $\Rbb^D$ and a partial order $\trianglelefteq$. The partial-order encodes the direction in cube spacetime and this is  provided by the \tsedef{cube-space rule}. 
Two points $x$ and $y$ lying in $\Rbb^d$ are only connected with $y$ coming \emph{later} than $x$, denoted as $x \trianglelefteq y$, if \emph{every} coordinate of $x$ is less than or equal to the corresponding coordinate of $y$. 
That is
\begin{equation}
	x \trianglelefteq y 
	\quad \Leftrightarrow \quad 
	x_i \leq y_i \quad \forall \; i \in \{1,2,\ldots,D\} 
	\, ,
	\;\; \forall \; x,y \in \Rbb^D \, .
	\quad
	\text{(cube-space rule)}
	\label{e:cubespacerule0}
\end{equation}
This $\trianglelefteq$ defines a partial order (see \appref{as:partialorder}) on cube-space. 

This cube-space rule generalises the notion of causality that lies behind the existence of a time coordinate and so an order in the points of traditional spacetimes. 
Note that any point $x$ is ``earlier'' than a point $y$, that is $x \trianglelefteq y$, only if $y$ lies in a hypercube with one corner at $x$, hence the name of the space and the rule.

\subsection{Distance}\label{s:Lpcsdistance}

We also need a measure of the separation $d(x,y)$ between two points $x,y \in \Rbb^D$ in our cube spacetime.
The reason is that we only want to connect to nodes which are less than a distance $R$ apart where $R$ will be one of our model parameters. This is often the only criteria used in traditional random geometric graph models \citep{G61,P03, DC02} but it is rarely imposed when creating DAGs. 

With normal geometry, Euclidean or more general Riemannian spaces, the space comes with a measure of distance which satisfies the criteria to be a distance in the formal mathematical sense (\defref{d:distance} in \appref{as:distance}). These distance measures often satisfy the triangle identity making them a metric in the formal mathematical sense (\defref{d:metric} in \appref{as:distance}). 
However, we are more interested in spacetime geometries with order, such as Minkowski spacetime and de Sitter spacetime. In these examples, the separation of points is given by the proper-time $\tau$ which is neither a distance nor a metric in the usual formal sense (see \appref{as:propertime}).

We choose to work with a more general definition of distance, one often used with data, the \tsedef{Minkowski distance}\footnote{Note that this distance is unrelated to distances used in Minkowski spacetime as described in \appref{as:propertime}.}, often referred to as the p-norm or the \Lp{}-norm. This measures the separation $\|x-y\|_p$ between two points $x$ and $y$, each with $D$ real coordinates $x_i$ and $y_i$  ($i\ \in\{1,2,\ldots,D\}$), as
\begin{equation}
	\|x-y\|_p 
	= \Big(\sum_{i=1}^D |x_i-y_i|^p\Big)^{1/p}  \, .
	\label{e:Lpdef0}
\end{equation}
Note we always take the positive real root so $\|x-y\|_p  \geq 0$.
With $p\geq 1$, this Minkowski distance $\|x-y\|_p$ satisfies the formal conditions for a distance (\defref{d:distance} in \appref{as:distance}). Further, it also satisfies the criteria to be a metric (\defref{d:metric} in \appref{as:distance}) as it also satisfies the triangle inequality. For instance, $p=2$ is the familiar Euclidean  distance measure which has special properties; it is invariant under rotations and translations. On the other hand, the \Lp{} metric is used for other values of $p \geq 1$ in data science when there is often no natural rotational symmetry in the data \cite{DD09}.

However, given the partial-order on our cube-space means that we have no use for the Minkowski distances $\|x-y\|_p$ except when $x \trianglelefteq y$. 
So we define what we call an \tsedef{ordered distance}{\footnote{This is very similar to, but not identical to, the \tsedef{time separation function} in Definition 2.8 of \citet{KS18} used for work on \tsedef{Lorentzian length spaces}.}  $d_p(x,y)$ between two points $x$ and $y$  where
\beq
	d_p(x,y) 
	=
	\begin{cases}
		\Big(\sum_{i=1}^D (y_i-x_i)^p\Big)^{1/p} &
		\;\; \text{if} \; x \trianglelefteq y \, ,
		\\
		\text{undefined} & \;\; \text{otherwise} . 
	\end{cases} 
	\label{e:dpdef}
\eeq
This defines what we call our \tsedef{\Lp{}-cube spacetime}, a triple $(\Rbb^D, \trianglelefteq, d_p)$ with points taken from $\Rbb^D$, with a partial order $\trianglelefteq$ defined using the cube-space rule, and an ordered distance $d_p$.  
We set out our formal framework for distance measures in spaces with order in \appref{as:odistancemetric}. 

The main reason for us working with this definition of distance $d_p(x,y)$ comes when we consider the triangle identities. For $p\geq 1$, our distance $d_p(x,y)$ satisfies the \tsedef{ordered triangle identity} \citep[see pp.~30--31 of][]{HLP52} that  
\beq
	d_p(x,z) \leq d_p(x,y) + d_p(y,z) 
	\, ,
	\quad  \text{where} \quad 
	x \trianglelefteq y \trianglelefteq z \, .
	\label{e:omtridef}
\eeq
This is weaker than the usual triangle identity for a normal metric since in this case the inequality only needs to hold for one permutation of the points $x$, $y$ and $z$. For instance we do not need $d_p(x,z) + d_p(y,z)$ to have any relation to $d_p(x,y)$. 
So for $p\geq 1$ the ordered distance $d_p(x,y)$ of \eqref{e:dpdef} is an \tsedef{ordered metric} following the definitions of \appref{as:odistancemetric}.

The primary result we want to exploit in our work is the fact that this ordered distance $d_p(x,y)$ of \eqref{e:dpdef} for both $0<p\leq1$ and $p<0$ also satisfies the \tsedef{reverse ordered triangle identity} \citep[see pp.~30--31 of][]{HLP52}
\beq
	d_p(x,z) \geq d_p(x,y) + d_p(y,z) 
	\, ,
	\quad  \text{where} \quad 
	x  \trianglelefteq y \trianglelefteq z \, .
	\label{e:ormtridef}
\eeq
That is for $p\leq1$ (excluding $p=0$) the ordered distance $d_p(x,y)$ of \eqref{e:dpdef} is an \tsedef{ordered reverse metric} from the definitions of \appref{as:odistancemetric}. This regime mimics the properties we have for the spacetimes used in relativity, such as Minkowski spacetime, since the proper time provides an example of an ordered reverse metric, see \eqref{ae:ormtridef} in \appref{as:odistancemetric}.

One point worth noting is that the case of $p=1$ is rather special. If two points are connected, say $x  \trianglelefteq y$, then all paths between those two points have the same length \citep[see pp.~30--31 of][]{HLP52}. This also means that the triangle inequalities are always satisfied with exact equality, that is $d_{p=1}(x,z) = d_{p=1}(x,y) + d_{p=1}(y,z)$ if $x  \trianglelefteq y \trianglelefteq z$.

\subsection{Geodesics}\label{s:Lpcsgeodesic}

Given the ordered distance $d_p$ of \eqref{e:dpdef} we can define the ordered geodesic as paths (a formal definition of a path is given in \secref{s:path} below) in this space which are extremal in length (see \appref{as:pggd} for a longer discussion of geodesics). The ordered geodesic in our \Lp{}-cube spacetime from the source node $s_i=0$  to the  sink $t_i=1$ is the diagonal $\gamma_\mathrm{O}  = \{ x | x_i = \lambda, 0 \leq \lambda \leq 1\}$. This is the unique ordered geodesic between these points except when $p=1$ when all paths are ordered geodesics. So we will use the diagonal as the reference path in our analysis below. 

The idea is that when we look at different paths in the DAG embedded in \Lp{}-cube spacetime, the ordered geodesic represents the most important path between two points in the continuous space, so we are asking what path in a DAG represents the `best' path in the continuum space in which the DAG is embedded. Choosing the diagonal as the best path between two points is a modelling assumption, but one that can be justified by the importance of ordered geodesics in several disciplines: in the formal mathematics of geometry, in physics (e.g.\ mechanics), computer science \citep[e.g.\ scheduling][]{KW59} and even social science \citep[e.g.\ bibliometrics in][]{HPEO24}.

\section{DAGs embedded in spacetime} \label{s:dagst}

In this section we look at a simple discrete model that creates a random DAG embedded in the continuous \Lp{}-cube spacetime.

\subsection{Adding order to random geometric graphs}

A useful baseline for studying networks with an underlying geometry is provided by random geometric graphs (RGGs). In these models, nodes of a network are placed at random in a continuous space and edges are added between nearby nodes. Most examples are for the case where the density of points is constant on average and where the space is Euclidean space \citep[for example][]{G61,P03, DC02}.
Equally, most existing work is for undirected networks with many cycles where there is no explicit sense of order in the network.

However, we are focusing on directed networks with no cycles, \tsedef{directed acyclic graphs} (DAGs). Such networks have an explicit sense of order since if you follow a walk in network, moving from node to node along an edge in the direction of that edge, you can never return to any node you pass through; there are no cycles. This lack of a return option means that every node in a DAG has a sense of direction, so following an edge is moving ``forward'' consistently. 

More formally, consider a directed network $\Gcal=(\Vcal,\Ecal)$ where $\Vcal$ is the set of nodes and $\Ecal$ is the set of directed edges. 
Then a \tsedef{path} $P$ in a directed network is a sequence of nodes $P = (v^{(a)})_{a=0}^\ell = (v^{(0)} \rightarrow v^{(1)} \rightarrow \ldots \rightarrow v^{(\ell)})$ in which consecutive pairs of nodes form directed edges in the correct orientation, that is $\dedge{v^{(a-1)}}{v^{(a)}} \in \mathcal{E}$ for all $a = 1, \ldots, \ell$.

The order present in a DAG is defined as $u \preceq v$ whenever there is a path from node $u$ to node $v$. That is 
\begin{equation}
    \forall \; u,v \in \Vcal \quad 
    \begin{cases}
    	u     \preceq v & \Leftrightarrow \quad  \text{a path from node }u\text{ to node }v\text{ exists,} \\
    	u \not\preceq v & \Leftrightarrow \quad \text{no path from node }u\text{ to node }v\text{ exists.}
    \end{cases}
    \label{e:preceqdef}
\end{equation}
The lack of cycles in a DAG ensures this is a `partial order' as defined in \appref{as:partialorder}.

So creating a random geometric DAG means we need to look beyond the vast majority of work on random geometric graphs and we have to start by working with a continuous space with a sense of order. As we reviewed in \secref{s:introduction}, the spacetimes used in relativity such as Minkowski space, have been used before \citep[for example][]{BB91,RW09,CGLE14,CE16,CE17} but these are too limited for our needs here. So instead we will use our \Lp{}-cube model defined in \secref{s:Lpcsmodel} with the cube-space rule of \eqref{e:cubespacerule0} providing the order.

\subsection{Limiting the length of edges}

Our second condition for an edge to exist in our random geometric DAG is to emulate the short distance criteria used in traditional random geometric graph models. Typically edges are only created between points less than $R$ apart, where $R$ is a positive real parameter of the model. Such a short length condition is rarely imposed when creating random DAGs embedded in spacetime; for example \citet{BB91,RW09,CGLE14,CE16,CE17} impose no such limit.

\subsection{The RGDAG model} \label{s:rgdag}

Using the cube spacetime rule and the length cutoff to define the directed edges between randomly placed  points gives us our \tsedef{random geometric DAG model} (RGDAG), a network embedded in our \Lp{}-cube spacetime. 

We start by placing $\Nbox$ nodes in the $D$-dimensional hypercube $[0,1]^D$.  Each of these nodes $v$ is given $D$ coordinates from the interval $[0,1]$, denoted as $v_i \in [0,1)$ with the index $i$ running over the $D$ coordinates, $i \in \{1, \ldots, D\}$. 
We chose the value of each coordinate $v_i$ uniformly at random from the interval\footnote{Since the choice of coordinates of each point is made independently of the coordinates of other points, this is very like a Poisson Point Process (PPP). The difference is that we fix the total number of points in the model rather than letting it be a random variable drawn from a Poisson distribution. However, for large $\Nbox$ the two processes will become indistinguishable. Also note that numerically, our coordinates are chosen from $[0,1)$ so that the value $1.0$ is excluded but this has no noticeable effect in our simulations.} $[0,1]$.  
In addition we add a global source $s_i=0$ and a global sink $t_i=1$ to give a total of $\Ntilde=(\Nbox+2)$ nodes.

Adding the edges is done using the cube-space rule of \eqref{e:cubespacerule0} and the ordered distance $d_p(x,y)$ of \eqref{e:dpdef}. Two nodes $u$ and $v$ are connected by a directed edge $\dedge{u}{v}$ from $u$ to $v$ if and only if
\begin{equation}
	\dedge{u}{v}
	\; \Leftrightarrow \;
	\begin{cases}	
		u_i \leq v_i \;\; \forall \; i \in \{1,2,\ldots,D\}, & \text{(cube-space rule),}
		\\
		0 < \Big(\sum_{i=1}^D (v_i-u_i)^p \Big)^{1/p} < R 
		\, , \;\; p \neq 0 \, ,  R>0 \, ,
	 	& \text{(distance cutoff).}
	\end{cases}	
\end{equation}
See for example \figref{fig:radius-square-comparison}. So the RGDAG model has two real parameters $p$ and $R$ chosen such that $R>0$ and $p \neq 0$. 

\begin{figure}[H]
  \centering
  \includegraphics[width=\textwidth]{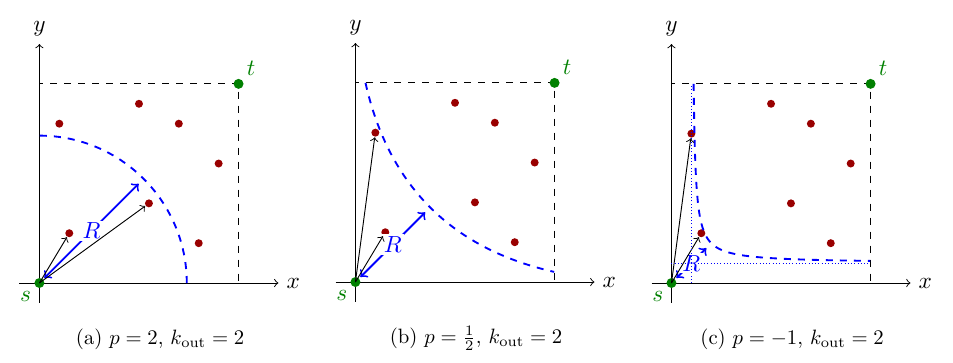}
  \caption{Illustration of the random geometric DAG construction in the unit square $[0,1]^2$ for three geometries. Nodes are scattered in the square, with a source node $s$ at the origin $(0,0)$ and a sink node $t$ at $(1,1)$. Directed edges are added from $s$ to all nodes lying within distance $R$ and with coordinates that are monotonically increasing in every dimension. Panel (a) shows the $p=2$ case, where the cutoff boundary is circular. Panel (b) shows the $p={1/2}$ case, where the cutoff boundary is given by the corresponding $p={1/2}$ curve. Panel (c) shows the $p={-1}$ case, where the forward connection region is unbounded.
  	  }
  \label{fig:radius-square-comparison}
\end{figure}
In addition, \tabref{t:dpproperties} summarises some of the properties of these spaces.

\begin{table}[H]
	\newlength{\tsewidth}
	\settowidth{\tsewidth}{$\triangle$ Identity}
	\centering
	\begin{tabular}{rl||ccc}
		\multicolumn{2}{l||}{Spacetime}      
		& \parbox[t]{\tsewidth}{\centering Curvature Unit Ball} 
		& \parbox[t]{\tsewidth}{\centering Ordered  $\triangle$ Identity} 
		& \parbox[t]{\tsewidth}{\centering Volume $|\Xcal^+(x,R)|$} 
		\\ \hline\hline
		\Lp{}-cube & $p > 1 $           & Convex  & Normal  & Finite \\
		\Lp{}-cube & $p =1 $            & Flat    & Equality& Finite \\
		\Lp{}-cube & $0 < p < 1 $       & Concave & Reverse & Finite \\
		\Lp{}-cube &$p < 0 $            & Concave & Reverse & $\infty$ \\\hline
		\multicolumn{2}{l||}{Minkowski} & Concave & Reverse & $\infty$ \\
	\end{tabular}
	\caption{Properties of various spacetimes. The curvature of the unit ball as viewed from the origin is given in each case. This is correlated with the ordered triangle identity, defined in \eqref{e:omtridef} for the normal case and in \eqref{e:ormtridef} for the reverse case. The forward connection kernel $\Xcal^+(x,R)$ is the region of space $\Xcal$ containing the points that can be connected to from a point $x$ when the maximum separation is $R$. The formula for the hyper-volume $|\Xcal^+(x,R)|$ is discussed in \appref{as:kout}. For Minkowski spacetime, the separation is measured using proper time $\tau$ of \eqref{e:propertimemst} and $R$ is the largest proper-time allowed between points.}
	\label{t:dpproperties}
\end{table}

At this stage, we have a DAG $\Dcal^\mathrm{(rgg)}$ with $\Ntotal$ nodes. However, not all nodes will be connected to the source and/or sink by at least one path, so this DAG $\Dcal^\mathrm{(rgg)}$ is not always an interval. So the last step is to find the interval DAG $\Dcal^\mathrm{(int)}(s,t)$ which is the subgraph where every point lies on at least one path from $s$ to $t$.

\subsection{Adding Heavy-Tailed Structure to Geometric DAGs}\label{s:citation-model}

Random geometric DAGs provide a useful baseline model because they allow us to isolate the role of the underlying geometry. However, they are also quite abstract. Real DAGs, such as citation networks and innovation networks, often have additional structural features which are not present in the RGDAG model, including heterogeneous degrees, clustered topic structure, and heavy-tailed degree structure \citep{P65b, P76, R98}. Our aim here is to look at DAGs produced in a model that retain some of the geometric structure of our RGDAG model but at the same time have a heavy-tailed degree structure so that they are more like DAGs found in real data. We will use the context of citation networks as the motivation for our model and draw on work to produce realistic models in that context.

We start with the DAG of our RGDAG model of \secref{s:rgdag} embedded in \Lp{}-cube spacetime. We proceed as before with $\Nbox$ nodes, each representing a paper with coordinates $v_i=(v_1,v_2)\in[0,1]^2$, each coordinate chosen uniformly at random between zero and one. 

However, for the citation network context, we view each paper $v$ in terms of two other coordinates: the \tsedef{time coordinate} $\tau(v)$ and \tsedef{topic coordinate} $\mu(v)$ where
\begin{equation}
	\tau(v) := \frac{v_1+v_2}{\sqrt 2} \, ,
	\qquad
	\mu(v) := \frac{v_1-v_2}{\sqrt 2} 
	\, .
	\label{eq:taumu}
\end{equation}
The $\tau(v)$ can be thought of as a publication date for each paper $v$ and we will ensure that there is an edge from node $u$ to node $v$ only if $\tau(u)<\tau(v)$.
The $R$ parameter effectively controls how far back a paper can cite and how close topic-wise it is to its bibliography.
The $\mu$ coordinate then represents some topic space. That is we will want papers with similar $\mu$ values to have a much higher chance of being connected as this represents documents writing on similar topics. Yet, time also plays a role. So two papers having the same $\tau$ coordinate but different $\mu$ coordinates are ``published'' at the same time so can never cite each other however close they are in terms of their topic coordinate $\mu$.
The edges between nodes and the citations between documents, will be constructed in two stages.

The first set of edges are going to represent what we call \tsedef{primary references}. For a each paper $v$, we add edges that represent entries in the bibliography of $v$ and which are recent papers in fields close to the main topic of paper $v$. This is done using the same \Lp{}-cube space rule used in our RGDAG model. 
That is we create an edge from any primary reference node $r$ to our current paper $v$ where $d_p(r,v) < R$ which ensures that $\tau(r)<\tau(v)$. The $R$ parameter of our model now sets the expected number of primary references per node, $\tseexpect{k_{r,\text{out}}}$, not the total number of references. 

Once we have the set of primary references of the node $v$, say $\Rcal(v) = \{r | d_p(r,v)<R\}$, we now add what we call \tsedef{secondary references}. To do this we look at every primary reference of $v$, say $r \in \Rcal(v)$ and we look at the references of those nodes which we call \tsedef{secondary references}, $r^\prime \in \Rcal(r)$.
With probability $\phi$, another parameter of this model, we add each secondary reference to the bibliography of node $v$, that is we add an edge $\dedge{r'}{v}$. 
Each secondary reference $r'$ has a path of hop length two, $(r' \rightarrow  r \rightarrow v)$ to the node $v$.
So the final expected degree is $(1+\phi)$ times the expected number of primary references per node as set by $R$. 
Essentially this second step mimics a process in which the authors of a new paper $v$ read recent papers, the primary references, and these form the core of their bibliography. They also read these primary references and read/copy some of the references, our secondary references, to the bibliography of the new paper $v$.  This copying mechanism is well known to generate fat-tailed degree distributions, both in terms of generic network modelling and citation networks \citep{KRRSTU00,V01a,ES05,GAE15a, KR05, SR03, PSGMC20}.

\begin{figure}[H] 
	\centering
	\includegraphics[width=0.5\linewidth]{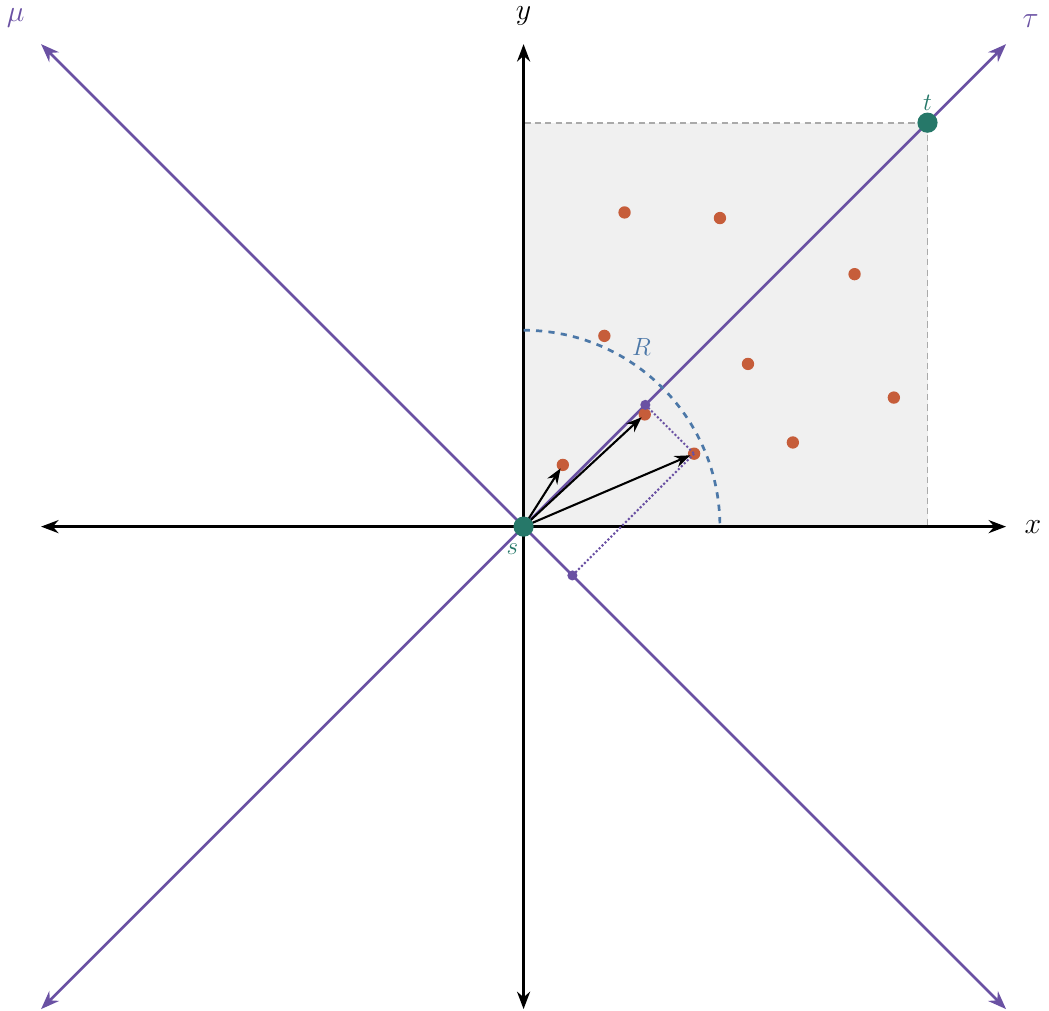}
	\caption{Illustration of the two-dimensional latent time--topic space. Nodes are placed in the original unit square, with global source $s=(0,0)$ and sink $t=(1,1)$. The rotated time axis $\tau$ lies along the source--sink diagonal and defines the temporal ordering, while the orthogonal axis $\mu$ represents the topic direction. The dashed blue curve shows the distance cutoff $R$ from the source, illustrated for $p=2$, and the black arrows indicate admissible forward connections. The dotted purple lines shows the decomposition of a node's displacement into its $\tau$ and $\mu$ coordinates.
	}
	\label{PARGDAG}
\end{figure}

\subsection{Paths in DAGs}\label{s:path}

Consider a directed network $\Gcal=(\Vcal,\Ecal)$ where $\Vcal$ is the set of nodes and $\Ecal$ is the set of directed edges. A path $P$ is a sequence of nodes $P = (v^{(a)})_{a=0}^\ell = (v^{(0)} \rightarrow v^{(1)} \rightarrow  \ldots \rightarrow  v^{(\ell)})$ in which consecutive pairs of nodes form directed edges in the correct orientation, that is $\dedge{v^{(a-1)}}{v^{(a)}} \in \mathcal{E}$ for all $a = 1, \ldots, \ell$.

The \tsedef{length}  $L(P)$ of a path $P$ is simply the sum of the lengths of each edge in the path.
\begin{equation}
	L(P) 
	= 
	L\big( (v^{(a)})_{a=0}^\ell \big) 
	= 
	\sum_{a=1}^{\ell} L\big(\dedge{v^{(a-1)}}{ v^{(a)}}\big)
	\, .
	\label{e:Ldef}
\end{equation}
We will study two different ways to assign a {length} to each edge and hence to any path. First we define the \tsedef{hop length}  $\Lhop(P)$ of a path $P$ to be the number of edges in the path, the number of nodes minus one. So the hop length of each edge is equal to one, $L(e)=1 \; \forall \; e \in \Ecal$ and 
\begin{equation}
	\Lhop(P) 
	= 
	L^\mathrm{(hop)} \big( (v^{(a)})_{a=0}^\ell \big) 
	= 
	\sum_{a=1}^{\ell} 1
	= 
	\ell \, .
	\label{e:Lhopdef}
\end{equation}

Our second measure of length assigns each edge $\dedge{u}{v}$ to have a length (or weight) equal to the ordered distance $d_p(u,v)$ of \eqref{e:dpdef} between the head and tail nodes of that edge. We define the \tsedef{weighted length} $\Lweight(P)$ of a path $P$ to be the sum of these ordered distance edge lengths
\begin{equation}
	\Lweight(P) = \Lweight\big( (v^{(a)})_{a=0}^\ell \big) = \sum_{a=1}^{\ell} d_p(v^{(a-1)},v^{(a)}) \, .
	\label{e:LWdef}
\end{equation}

In many network settings, the geodesic in a network is identified with the shortest path, especially in undirected and unweighted graphs, see ch.10 in \citet{C21} for example. This convention is natural in many contexts such as when asking how many steps separate two individuals in a social network (``six degrees of separation'') or finding the fastest route between two locations. This emphasis on short paths is central to the theory of small-world networks, where rapid reachability between otherwise distant nodes is one of the defining features. So one path we will study is the hop shortest path $P^{(\mathrm{hop})}_{\min}$ from $s$ 
to $t$, a path from global source $s$ to global sink $t$ with the least number of edges, so formally
\begin{equation}
    P^{(\mathrm{hop})}_{\min} = \argmin_{P \in \mathcal{P}(s,t)} \Lhop(P) \, .
    \label{e:Phopmin} 
\end{equation}
Here $\mathcal{P}(s,t)$ is the set of all directed paths from source $s$ to sink $t$.
Likewise, we will study the shortest weighted path $P^{(W)}_{\min}$ where
\begin{equation}
    P^{(W)}_{\min} = \argmin_{P \in \mathcal{P}(s,t)} L^{(W)}(P) \, .
    \label{e:PWmin}
\end{equation} 

For directed acyclic graphs, however, the focus on shortest paths alone is no longer justified. A DAG carries an intrinsic ordering, and paths are not merely routes between two points; they can represent ordered chains of dependency, influence, or causal precedence. In such settings, the most informative path is not necessarily the shortest one. In scheduling problems, the longest path gives the minimum time needed for a process \citep{CG72,LTW63,LTW63a}. In studies of quantum gravity, the longest path of a DAG embedded in a physical spacetime represents the best approximation to a geodesic \citep{BB91,BG91,RW09}. In an innovation network, for example, if a scientific paper enables a patent, which later enables a technology, which eventually contributes to a vaccine, then the intermediate steps are not redundant, they are part of the causal chain by which the final object comes into existence \citep{HPEO24,HPEO24a,PE25}. Citation networks provide a similar example, where long paths can represent extended intellectual lineages rather than inefficient detours \citep{CGLE14,CE16,CE17,ECV20}. 

So with this in mind we will also look at the paths from the global source $s$ to the global sink $t$ which have the longest hop length from 
$P^{(\mathrm{hop})}_{\max}$ and the longest weighted path $P^{(W)}_{\max}$, where
\begin{eqnarray}
    P^{(\mathrm{hop})}_{\max} 
    &=& 
    \argmax_{P \in \mathcal{P}(s,t)} \Lhop(P) \,,
    \label{e:Phopmax} 
    \\
    P^{(W)}_{\max} 
    &=&
     \argmax_{P \in \mathcal{P}(s,t)} L^{(W)}(P) \, .
     \label{e:PWmax}
\end{eqnarray}

Note that in a network there are paths which satisfy the criteria of an ordered geodesic as defined in \appref{as:odistancemetric}. In a DAG, the shortest and longest paths using any given measure of edge length always define ordered geodesics, minimal ordered geodesics and maximal order geodesics respectively. So when we make comparisons between the ordered geodesic in the continuous \Lp{}-cube spacetime and the shortest/longest hop/weighted paths in the embedded DAG, we are making comparisons between ordered geodesics of various types. However, when we refer to ``the ordered geodesic'' we will mean the ordered geodesic defined in the continuous \Lp{}-space. The paths used in the networks will be referred to by the names already given.

\section{Comparing paths in RGDAG and \Lp{}-cube spacetime} \label{s:results}

So far we have established the geometry for ordered spacetimes, both the continuous \Lp{}-cube spacetime and the discrete RGDAG. Now we can ask how the paths in the discrete space are related to those in the continuous space. 
Instead of assuming that the shortest path in a network is the correct analogue of a geodesic, we ask which path in a DAG best approximates the continuum reference path. 

The \Lp{}-cube spacetime construction introduced above is useful precisely because it is ordered and yet it provides a controlled way to compare two regimes. Looking at the obvious representation of this spacetime, using orthogonal coordinates with identical scales for each coordinate axis, our intuition is that the most important path from source to the sink is the straight line, the diagonal of the unit hypercube $[0,1]^D$. 

We will not investigate the concept of ``straightness'' here. 
Rather, we note that this diagonal is also an extremal path in terms of the ordered distance $d_p$ of \eqref{e:dpdef} so this diagonal is indeed an ordered geodesic in our terminology. A more physical motivation for focusing on this diagonal is that in two-dimensions this cube-space is identical to the $1+1$-dimensional Minkowski spacetime. In that case, the diagonal is the path followed by any massive non-interacting particle moving from the spacetime point $(0,0)$ to the spacetime point $(1,1)$. That is, it is the path of least effort in that physical example.

However, the true power of our \Lp{}-cube spacetime model comes when we consider different values of $p$. For $p>1$, the diagonal is the shortest continuum source-to-sink path in terms of our ordered distance measure $d_p$ of \eqref{e:dpdef}. In this case, the diagonal is a minimal ordered geodesic.
However, as we have shown, our ordered version $d_p$ of \eqref{e:dpdef} for $0<p<1$ and $p<0$ is a reverse ordered metric satisfying the reverse ordered triangle identity \eqref{e:ormtridef}. Now we find that the diagonal has the longest length of any path between the source and sink as measured by $d_p$ of \eqref{e:dpdef}. In this case the diagonal is a maximal ordered geodesic in our terminology.

What we now want to do is compare the diagonal, the ordered geodesic in the continuum case, the various paths in the discrete RGDAG embedded in the \Lp{}-cube spacetime. Which path in the DAG is the best approximation to the ordered geodesic in the continuum model? To do this we first define measures that allow us to compare different paths.

\subsection{Geodesic Similarity Measures} \label{sec:similarity}

The simplest measure we can define is based on the weighted length defined by the $\Lp{}$ metric or our ordered distance $d_p$ of \eqref{e:dpdef}. In these units the length of the side of our hypercube $[0,1]^D$ is simply $1.0$.
The weighted length of our reference path between source $s$ and sink $t$, the diagonal $\gamma_\mathrm{O}$ is used to define the ordered geodesic, so $\geolength = d_p(s, t)$.  For our source at $s_i=0$ and sink at $t_i=1$ we have that $\geolength=D^{1/p}$.

To compare the weighted length of any path $P$ in our DAG, we use the same ordered distance measure $d_p$ to find the length of each edge and the sum gives us the weighted length $L_p(P)$ of a path $P=(v^{(a-1)})_{a=0}^\ell$
\begin{equation}
    L_p(P) = \sum_{a=1}^{m} d_p(v^{(a-1)},  v^{(a)}).
\end{equation}
We compare this to the length of the ordered geodesic using the normalised length deviation $\Delta_L(P)$ where
\begin{equation}
    \Delta_L(P) := \frac{L_p(P)}{\geolength} - 1.
    \label{e:DeltaLdef}
\end{equation}
Thus $\Delta_L(P) = 0$ if the path has the same total geometric length as the continuum reference, $\Delta_L(P) > 0$ if it is longer, and $\Delta_L(P) < 0$ if it is shorter. Note we use the weighted length as the basis for our metric even if the path was chosen by optimising the number of hops in the path. The hop distance is an intrinsically discrete measure which does not have exact analogue which we can use for the ordered geodesic curve defined in the continuum space\footnote{Other options include using an approximation where we use $R$ or the inter-point length scale $(\rho)^{1/D}=N^{1/D}$ as the units for the measurement of the length of the ordered geodesic in the continuum.}.

To measure how close a path $P$ follows our reference curve, to the continuum ordered geodesic $\gamma_\mathrm{O}$, we define the point-to-reference curve distance $d_\bot(v,\gamma_\mathrm{O})$ using the Euclidean perpendicular distance\footnote{We might want to use the same Minkowski distance as we use for the weighted length of paths. However, all we require here is some measure of the deviation of the discrete path from the diagonal so for simplicity we just use Euclidean geometry for this measure. In any case, there is often no ordered distance defined between the path and the point on the continuum ordered geodesic that we want to consider as closest e.g.\ the point on the ordered geodesic with the same time parameter $\tau$ as the node in the DAG. So we have to use a metric not an ordered metric for this measure.}
\begin{equation}
    d_\bot(v,\gamma_\mathrm{O}) 
    := 
    \min_{u \in \gamma_\mathrm{O} } \left( \sum_{i=1}^{D} |v_i - u_i|^2 \right)^{1/2} 
    \label{e:perpdist}
\end{equation}
where $v$ is the point and $\gamma_\mathrm{O}$ is the reference path of interest.
In our case we apply this to each point in a discrete path $P$ in our DAG and then take the mean to give the perpendicular distance measure $E_\bot(P)$. 
That is
\begin{equation}
    E_\bot(P) := \frac{1}{|P|+1} \sum_{v \in P} d_\bot(v,\gamma_\mathrm{O}) \, . 
    \label{e:DeltaEdef}
\end{equation}
The quantity $E_\bot(P)$ is zero if and only if every node of the path lies on the reference path $\gamma_\mathrm{O}$.

\subsection{Results for the random geometric DAG}

\subsubsection{Positive $p$-Values}\label{s:posp}

We start by looking at how the various paths behave in the two positive $p$ regimes, $0<p<1$ and $p>1$, in dimensions $D = 2, 3, 4$. The analysis is organised by mean out-degree $\tseexpect{\kout}$, which is varied by changing the geometric cutoff $R$ (see \appref{as:kout} for an analytic expression).

The results for $\Delta_L(P)$ are shown \figref{fig:summary-length-deviation} while those for $E_\bot(P)$ are given in \figref{fig:summary-perpendicular-distance}. A few features stand out.

First, the longest paths in the DAG, by hop or weighted length, are better approximations to the ordered geodesic when $p<1$ but the shortest paths in the DAG, by hop or weighted length, are better approximations to the ordered geodesic when $p>1$. Thus the optimal path in the DAG satisfies the same triangle identity as the underlying geometry used to create the DAG.

Secondly, for any given $p$, the best of the two optimal paths by hop length (shortest for $p>1$, longest for $p<1$) are poorer approximations to the ordered geodesic than the corresponding best path optimised for weighted length. This is necessary for the measure $\Delta_L(P)$ which uses the weighted length \Lp{} given that the weighted discrete path was defined by optimising the same quantity. What is more informative is that the perpendicular distance measure $E_\bot(P)$ shows the same result even though the perpendicular distance is not directly optimised when finding the longest/shortest paths in the DAGs.

Thirdly, for any given $p$, the best of the two optimal paths by hop length (shortest for $p>1$, longest for $p<1$) is still a good approximation to the ordered geodesic even if it is slightly less successful than the weighted path in following the ordered geodesic. This is important because if we only have the DAG and not the underlying space, e.g.\ a DAG created from a dataset, we can define the optimal hop length paths but not the weighted paths. Yet we can be confident that the optimal hop-length paths in the network are good choices as ordered geodesics in the discrete network.

Lastly, results for the optimal hop length paths, especially the shortest hops path, show significant fluctuations. This is because we are finding paths with small integer hop length. There can be many of these tied with the same value. In this case, during path reconstruction we select uniformly at random among all next nodes that preserve the optimal hop count. In addition, when the \Lp{} length of the ordered geodesic is slightly above an integer number of $R$, that is $0 < L_p(P) \mod R \ll 1$, the length of the shortest hops path may change sharply between different instances.

\begin{figure}[H]    
	\centering    
	\includegraphics[width=\linewidth]{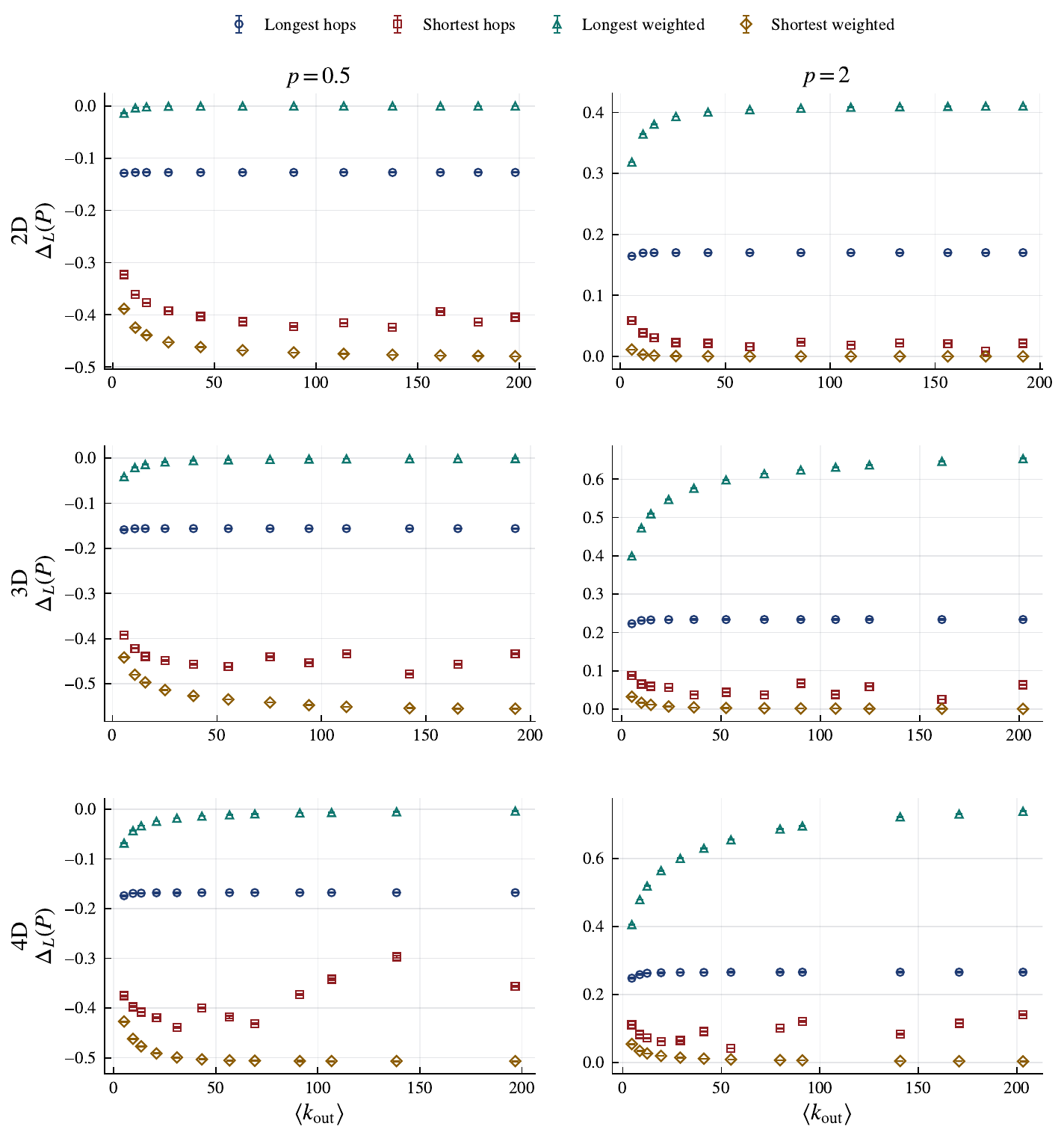}    
	\caption{Normalised length deviation $\Delta_L(P)$ for the four 
		graph-based paths as a function of mean out-degree $\langle \kout\rangle$, shown 
		for $p = 0.5$ and $p = 2$ in dimensions $D = 2, 3, 4$. Each point is averaged 
		over $500$ realisations for $D=2$ and $100$ realisations for $D=3,4$, with 
		$\Ntilde = 10{,}000$ nodes per; error bars show the standard error of the mean of $\Delta_L(P)$ across realisations. In the metric regime $p = 2$, $P^{(W)}_{\min}$ (yellow diamonds) remains closest 
		to zero across all dimensions, while in the fractional regime $p = 0.5$ this 
		role is taken by $P^{(W)}_{\max}$ (green triangles).}
	\label{fig:summary-length-deviation}
\end{figure}

\begin{figure}[H]    
	\centering    
	\includegraphics[width=\linewidth]{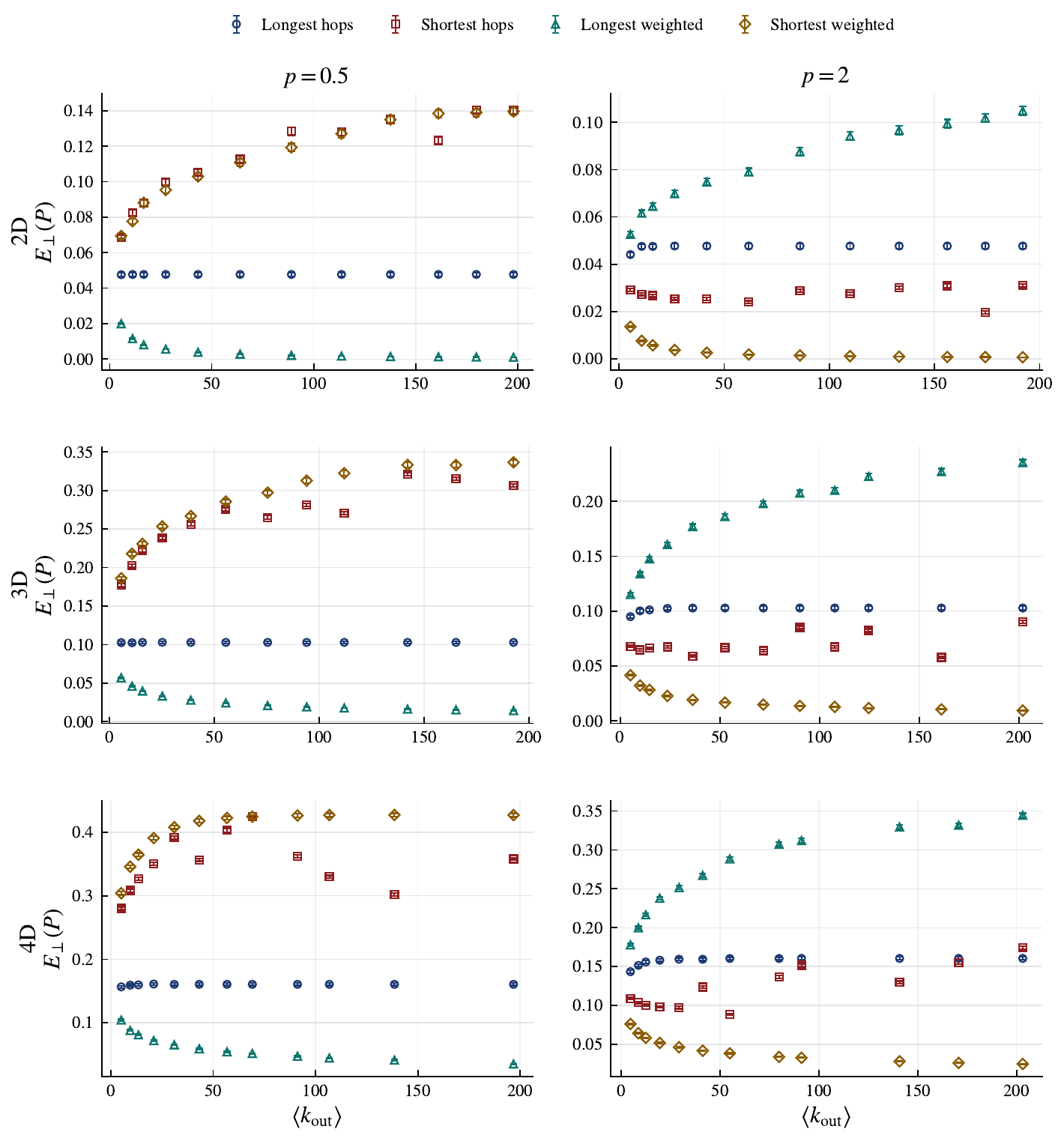}    
	\caption{Perpendicular distance measure $E_\bot(P)$ from the graph path 
		to the reference diagonal, plotted against mean out-degree $\tseexpect{\kout}$ 
		for $p = 0.5$ and $p = 2$ in dimensions $D = 2, 3, 4$. Each point is averaged 
		over $500$ realisations for $D=2$ and $100$ realisations for $D=3,4$, with 
		$\Ntilde = 10{,}000$ nodes per; error bars show the standard error of the mean of $E_\bot(P)$ across realisations. The same qualitative transition as in \figref{fig:summary-length-deviation} is 
		observed: $P^{(W)}_{\min}$ (yellow diamonds) is closest to the reference curve for $p = 2$, 
		whereas $P^{(W)}_{\max}$ (green triangles) is closest for $p = 0.5$.}   
	\label{fig:summary-perpendicular-distance}
\end{figure}

\clearpage

\subsubsection{Negative $p$-Values}\label{s:negp}

The properties listed in \tabref{t:dpproperties} show that the geometry of the RGDAG model in the $p<0$  regime is the closest match to that seen in Minkowski spacetime. As our analysis shows, this regime is perfectly well defined in our directed geometry and the ordered distance $d_p$ of \eqref{e:dpdef} is a well defined ordered reverse metric (as defined in \eqref{ae:ormtridef} of \appref{as:odistancemetric}) just as it was for $0<p<1$. Yet negative $p$ values have received almost no attention. So this section also acts as a proof of concept that this regime can be used for analysis. Again, we vary $R$ but we express our results in terms of the value for $\tseexpect{\kout}$ measured in our simulations (see \appref{as:kout} for analytic expressions).

Given our observation that the underlying geometry was the determinant in finding the paths closest to the ordered geodesic, we should expect that the $p<0$ regime should have the same properties as the $0<p<1$ regime given both obey the same reverse ordered triangle identity. This is exactly what we find as shown in \figref{fig:length-deviation-pm1} and \figref{fig:perp-distance-pm1}.

\begin{figure}[H]
    \centering
    \includegraphics[width=0.7\linewidth]{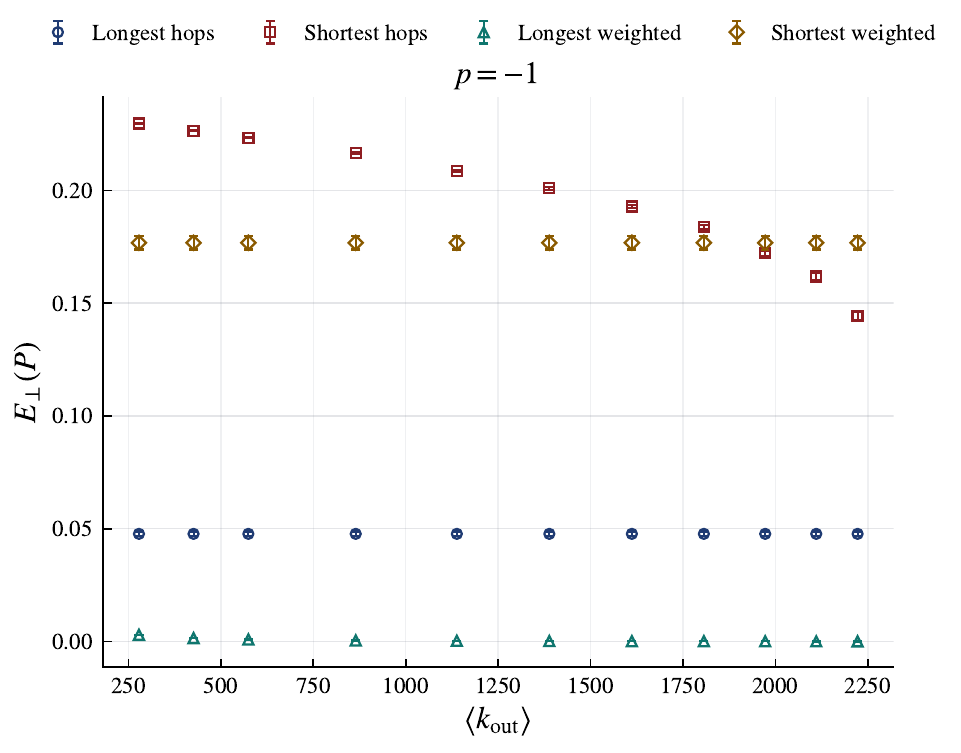}
    \caption{Perpendicular distance measure $E_\bot(P)$ from the graph path to the reference diagonal, plotted against mean out-degree $\tseexpect{\kout}$ for $p = -1$ in $D=2$ dimensions, with $\Ntilde = 10{,}000$ nodes, averaged over $500$ realisations per radius. Error bars show the standard error of the mean of $E_\bot(P)$ across realisations. As in the $0 < p < 1$ regime, $P^{(W)}_{\max}$ (green triangles) remains the path family closest to the reference, while the other path families remain substantially farther away.}
    \label{fig:length-deviation-pm1}
\end{figure}

\begin{figure}[H]
    \centering
    \includegraphics[width=0.7\linewidth]{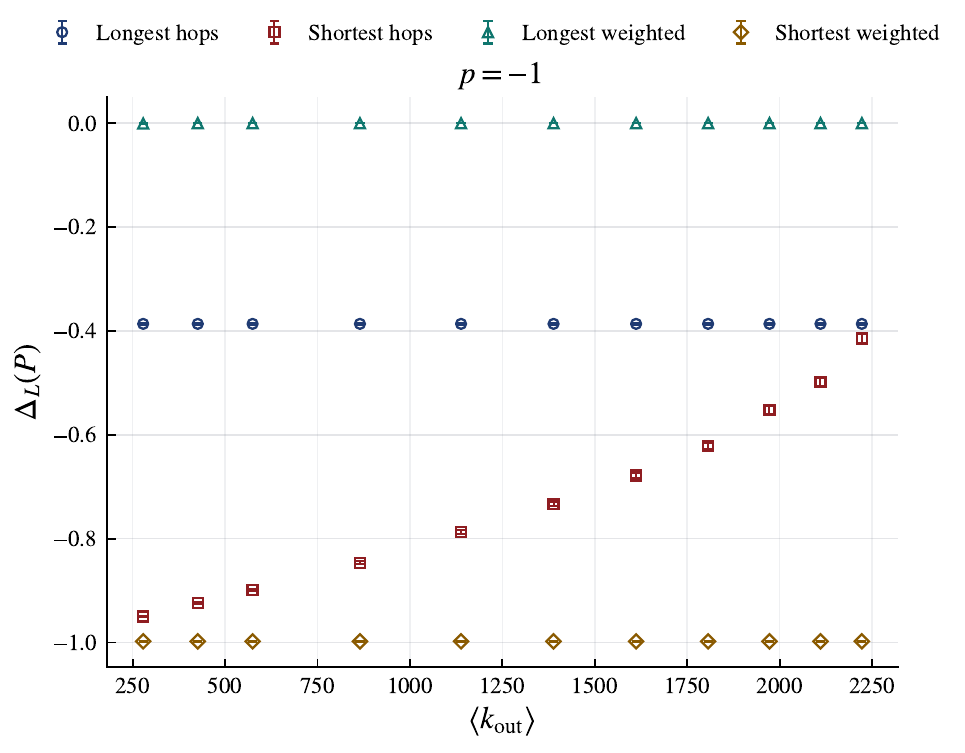}
    \caption{Normalised length deviation $\Delta_L(P)$ for the four graph-based paths as a function of mean out-degree $\tseexpect{\kout}$, $p = -1$ in $D=2$ dimensions, with $\Ntilde = 10{,}000$ nodes, averaged over $500$ realisations per radius. Error bars show the standard error of the mean of $\Delta_L(P)$ across realisations. As in the $0 < p < 1$ regime, $P^{(W)}_{\max}$ (green triangles) remains the path family closest to the reference, while $P^{(W)}_{\min}$ (yellow diamonds) deviates most strongly, consistent with the reverse ordered triangle identity shared by both regimes.}
    \label{fig:perp-distance-pm1}
\end{figure}

\subsubsection{The transition around $p=1$}

There is a change in the behaviour of our models at $p = 1$, the point at which the ordered distance $d_p$ of \eqref{e:dpdef} changes from an ordered metric obeying the ordered triangle identity \eqref{e:omtridef} for $p \geq 1$ to a reverse ordered metric obeying the reverse ordered triangle identity of \eqref{e:ormtridef}. 
This is reflected in the fact that the diagonal is a minimal ordered geodesic for $p>1$ but it is a maximal ordered geodesic for $0<p<1$ and $p<0$. 

At $p=1$ we are using the Manhattan metric to give our weighted distances and this has the special property that \emph{all} paths from source to target have the same weighted length and this includes the ordered geodesic. To examine this boundary we perform a finer sweep in $p$ around $p = 1$ and track both the length of paths through $\Delta_L(P)$ \eqref{e:DeltaLdef} and the distance to the ordered geodesic through $E_\bot(P)$ \eqref{e:DeltaEdef}. Results are shown in \figref{fig:fine-p-sweep-length-ratio} and \figref{fig:fine-p-sweep-perpendicular-distance}.

In the case of the paths with optimal weighted length, the crossover is especially clear. The two optimal weight paths meet at $p = 1$, as they must, with both varying linearly on either side of $p=1$ but with a discontinuity in slope at $p=1$.

The paths which optimise hop length show a much weaker dependence on the change in geometry. Hop length is discrete, and many paths can attain the same optimal value. 
When building an optimal hop-length path, we choose uniformly at random from the set of nodes that we can add to the path being built and which lie on an optimal hop count path. 
In $\Delta L(P)$, the two hop-based curves meet at $p=1$, where all paths have the same weighted length, but there is no corresponding exchange in which hop path lies closest to the continuum reference within this narrow range: the shortest-hop path remains closer to $\Delta L(P)=0$ on both sides of $p=1$. Similarly, in $E_\bot(P)$, no crossover is visible: the longest-hop path remains closer to the reference diagonal on both sides of $p=1$.

The weighted-path transition at $p=1$ is consistent across all dimensions studied, $D=2,3,4$. The hop-based paths do not show a corresponding exchange of the preferred path within this narrow range of $p$ in either $\Delta L(P)$ or $E_\perp(P)$.

\begin{figure}[H]
	\centering
	\includegraphics[width=\linewidth]{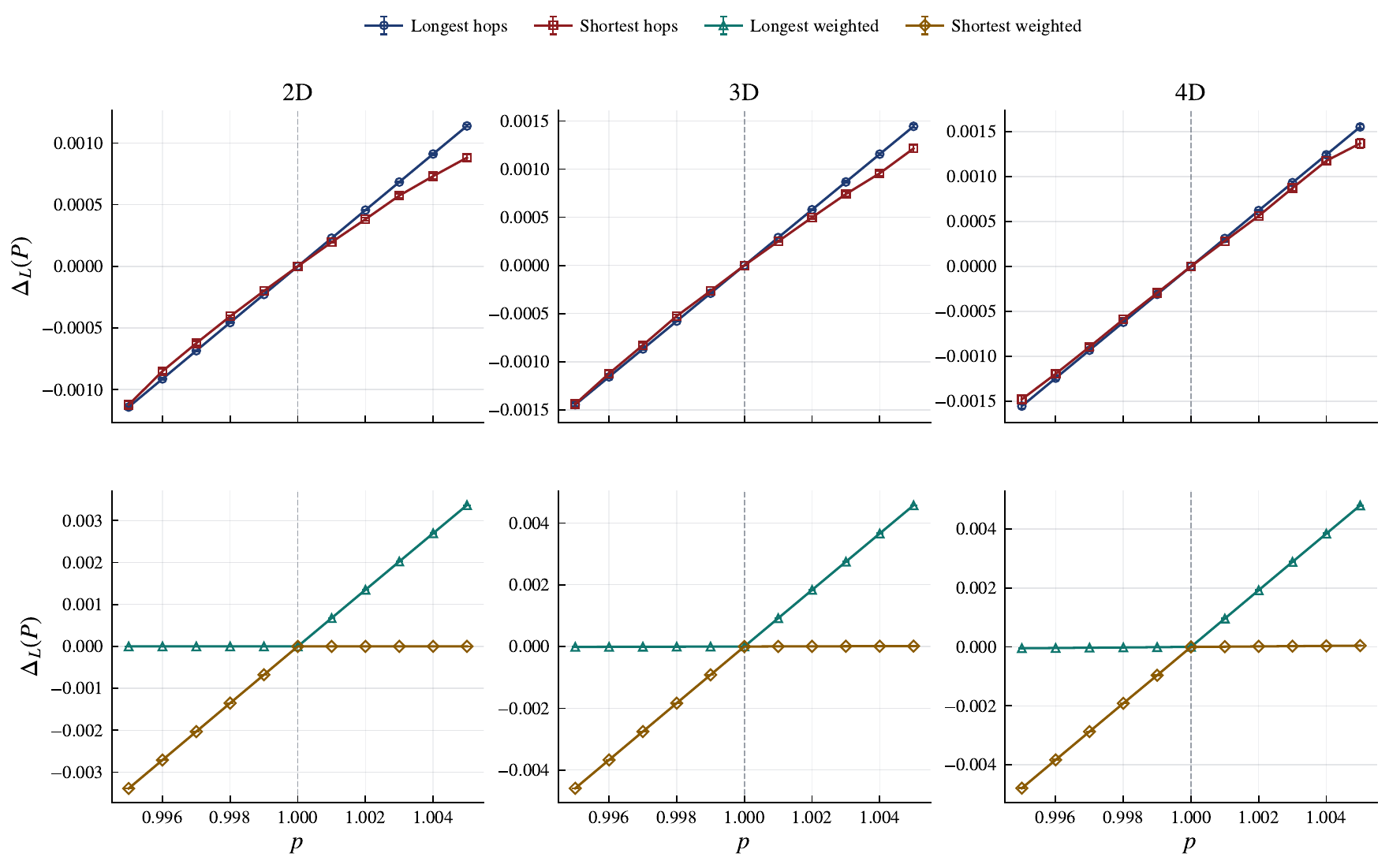}
	\caption{A sweep in $p$ around the transition at $p = 1$, showing the normalised length deviation $\Delta_L(P)$ of \eqref{e:DeltaLdef} for hop-length optimised paths (top row) and optimised weight paths (bottom row) in dimensions $D = 2, 3, 4$. This is for networks with $\Ntilde = 10{,}000$ nodes and $\tseexpect{\kout}=100$. 
		Error bars show the standard error of the mean of $\Delta_L(P)$ across realisations: $500$ realisations for $D=2$ and $100$ for $D=3,4$. 
		The longest and shortest weighted paths meet at $p = 1$ and exchange roles across the transition: $P^{(W)}_{\max}$ (green triangles) is 
		closest to the reference for $p < 1$, while $P^{(W)}_{\min}$ (yellow diamonds) is closest for 
		$p > 1$. The hop-based curves also meet at $p=1$, but do not show the same exchange in which path lies closest to the reference.
        The lines between data are added as an aid to visualisation and should not be interpreted as implying continuity or meaningful intermediate values between the observations.
		}
	\label{fig:fine-p-sweep-length-ratio}
\end{figure}

\begin{figure}[H]
	\centering
	\includegraphics[width=\linewidth]{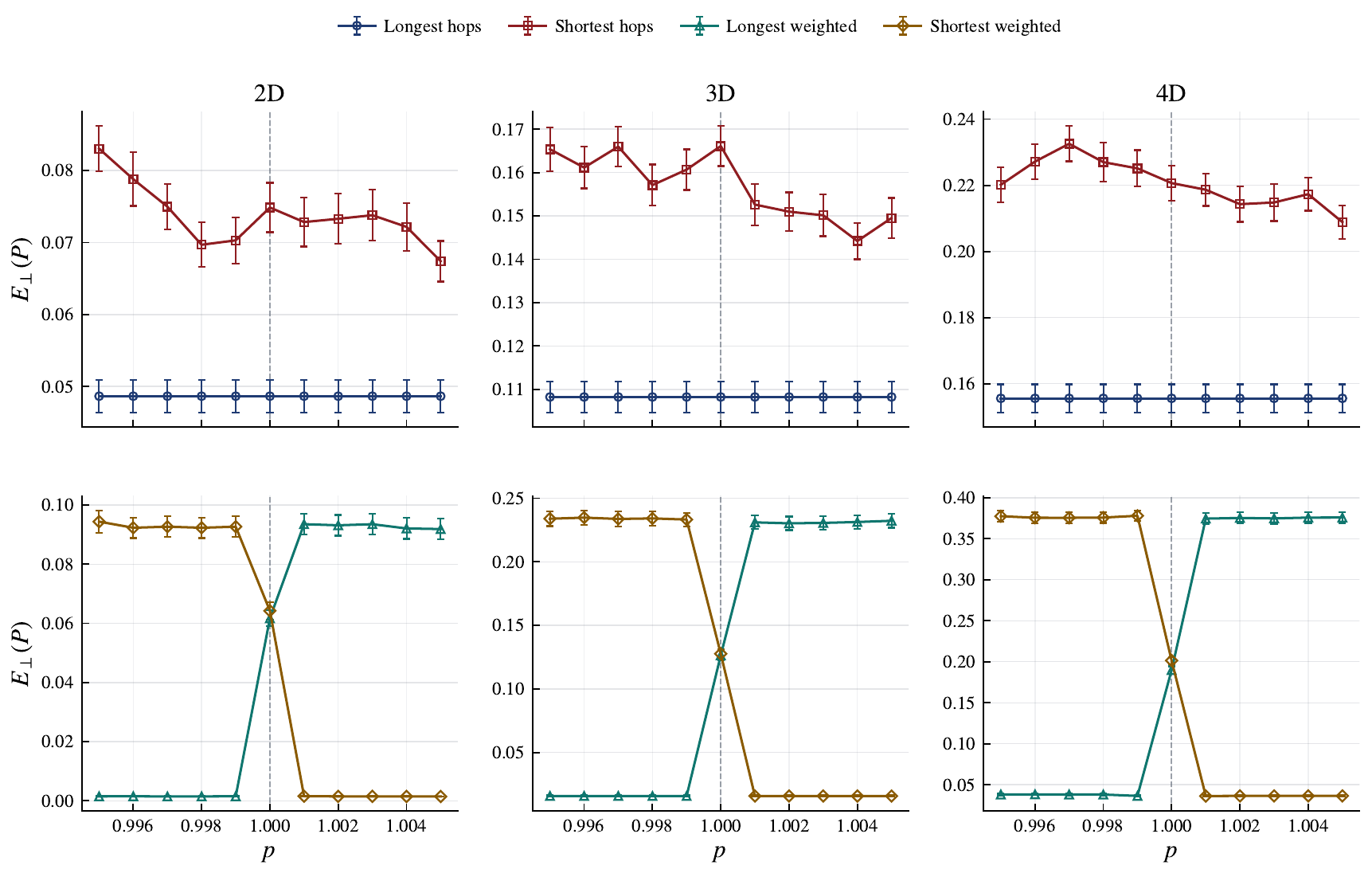}
	\caption{Fine sweep in $p$ around the transition at $p = 1$, showing the 
		perpendicular distance measure $E_\bot(P)$ for hop-based and weighted paths in 
		dimensions $D = 2, 3, 4$, with $\Ntilde = 10{,}000$ nodes and $\tseexpect{\kout}=100$. Error bars show the standard error of the mean of $E_\bot(P)$ across realisations: $500$ realisations for $D=2$ and $100$ for $D=3,4$. The weighted paths display a sharp crossover at 
		$p = 1$, with $P^{(W)}_{\max}$ (green triangles) closest to the reference for $p < 1$ and 
		$P^{(W)}_{\min}$ (yellow diamonds) closest for $p > 1$. The hop-based paths do not detect the same 
		transition. The lines between data are added as an aid to visualisation and should not be interpreted as implying continuity or meaningful intermediate values between the observations.
		}
	\label{fig:fine-p-sweep-perpendicular-distance}
\end{figure}

\subsection{Results for DAGs with Fat-Tails }\label{s:PA-model-results}

\Figref{fig:delta_L-PA} and \figref{fig:delta_E-PA} show that adding the preferential-attachment layer to our RGDAG model of \secref{s:rgdag} does not alter the qualitative relationship between the latent geometry and the graph path that best approximates the continuum ordered geodesic. 

For $p=2$, the shortest weighted path $P^{(W)}_{\min}$ remains closest to the reference according to both observables. Its normalised length deviation $\Delta_L(P)$ and perpendicular distance measure $E_\bot(P)$ approach zero as the mean base out-degree $\langle k_{r,\mathrm{out}}\rangle$ increases. The longest weighted path $P^{(W)}_{\max}$ shows the largest deviations, while the two hop-based paths generally lie between the weighted extremes.

For $p=0.5$, this ordering is reversed. The longest weighted path $P^{(W)}_{\max}$ has a length deviation $\Delta_L(P)$ close to zero and the smallest $E_\bot(P)$ across the full range of connectivity studied. The remaining path families have systematically negative length deviations, with $P^{(W)}_{\min}$ lying furthest below the continuum reference length. Thus, despite the addition of causally valid edges that are not required to satisfy the geometric radius cutoff, the distinction between the metric and fractional regimes is preserved. The weighted-path result is therefore robust to the addition of preferential-attachment structure.
\begin{figure}[H]
    \centering
    \includegraphics[width=\linewidth]{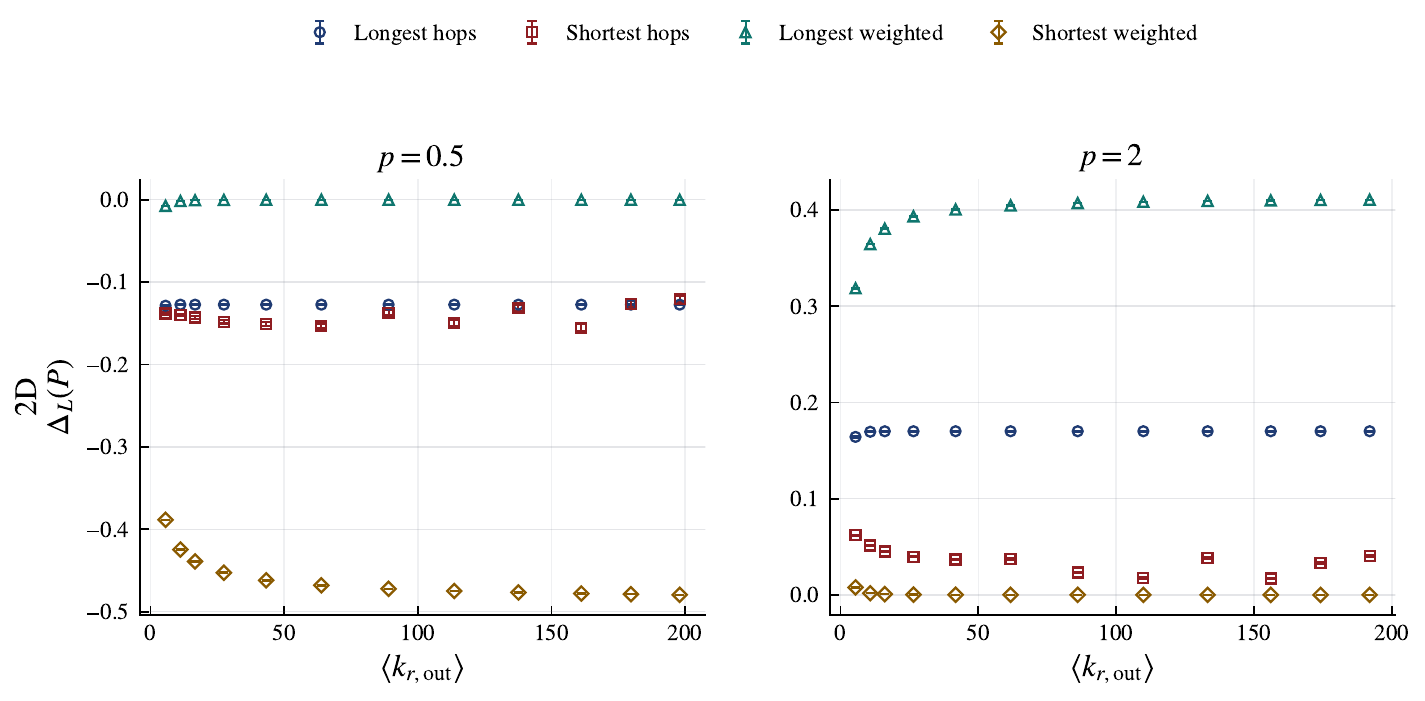}
    \caption{Normalised length deviation $\Delta_L(P)$ for the \Lp{}-cube space model extended with an extra edge copying step. We show the four optimal paths as a function of the expected primary out-degree $\tseexpect{k_{r,\text{out}}}$, shown for $p = 0.5$ and $p = 2$ in $D = 2$ dimensions. Each point is averaged over $500$ realisations, with $\Ntilde = 10{,}000$ nodes per graph; error bars show the standard error of the mean of $\Delta_L(P)$ across realisations. In the metric regime $p = 2$, $P^{(W)}_{\min}$ (yellow diamonds) remains closest to zero, while in the fractional regime $p = 0.5$ this role is taken by $P^{(W)}_{\max}$ (green triangles).}
    \label{fig:delta_L-PA}
\end{figure}

\begin{figure}[H]
    \centering
    \includegraphics[width=\linewidth]{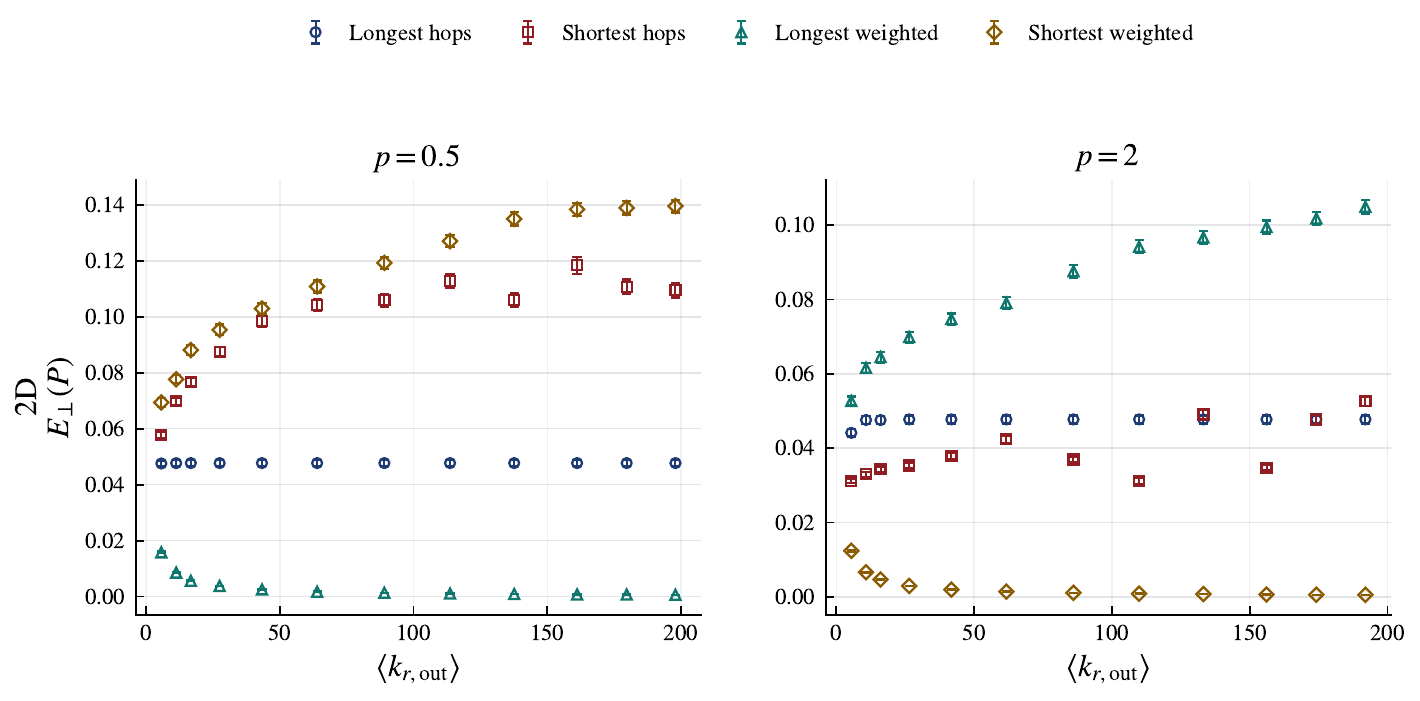}
    \caption{Perpendicular distance measure $E_\bot(P)$ for the \Lp{}-cube space model extended with an extra edge copying step. We show the four optimal paths as a function of the expected primary out-degree $\tseexpect{k_{r,\text{out}}}$, shown for $p = 0.5$ and $p = 2$ in $D = 2$ dimensions. Each point is averaged over $500$ realisations, with $\Ntilde = 10{,}000$ nodes per graph; error bars show the standard error of the mean of $E_\bot(P)$ across realisations. The same qualitative transition is observed: $P^{(W)}_{\min}$ (yellow diamonds) is closest to the reference curve for $p = 2$, 
	whereas $P^{(W)}_{\max}$ (green triangles) is closest for $p = 0.5$.}
    \label{fig:delta_E-PA}
\end{figure}

\section{Conclusion and Discussion} \label{s:conclusion}

In this paper, we asked whether the geometry in which a directed acyclic graph is embedded can be identified from its source--sink path structure. To do this we establish a set of concepts such as ordered distance and ordered geodesic appropriate for discrete and continuous spacetimes, a space where there is an explicit partial order on the points in the space. Then we constructed a family of ordered spaces, the \Lp{}-cube spacetimes, consisting of points $\mathbf{x}\in[0,1]^D\subset\Rbb^D$, with the cube-space rule \eqref{e:cubespacerule0} defining the partial order $\trianglelefteq$ and the ordered distance $d_p(\xvec,\yvec)$ defining the separation between ordered points. The single parameter $p$ controls the relevant triangle identity and therefore a key geometric property of the space: for $p>1$, $d_p$ satisfies the ordered triangle identity, while for $0<p<1$ and $p<0$ it satisfies the reverse ordered triangle identity, with $p=1$ marking the transition between the two regimes. This provides a family of continuous ordered spaces in which we can investigate how the underlying geometry is reflected in the paths of DAGs embedded within them.

We generate random geometric DAGs embedded in these ordered spaces, and so representing discrete approximations of the embedding space. We then compared four paths from the source node to the sink node: the shortest and longest weighted paths, $P^{(W)}_{\min}$ and $P^{(W)}_{\max}$, which is the summation of the \Lp{} distance between every consecutive pair of nodes on the path, and the shortest and longest hop paths, $P^{(\mathrm{hop})}_{\min}$ and $P^{(\mathrm{hop})}_{\max}$, which is the number of edges in the path. We compared these paths with the continuum ordered geodesic using two measures: the deviation of the weighted path length from the continuum ordered geodesic length, $\Delta_L(P)$, and the Euclidean perpendicular distance from the continuum ordered geodesic, $E_\perp$. We then extended this analysis to a citation-style model which retains the \Lp{}-cube spacetime embedding while adding a more realistic fat-tailed degree distribution through a local search (from which preferential attachment emerges), and this introduces additional edges beyond the strict geometric cutoff.

For the weighted paths, the path closest to the continuum ordered geodesic is the one whose extremal property agrees with the triangle identity of the underlying ordered space. For $p>1$, the shortest weighted path $P^{(W)}_{\min}$ is closest to the ordered geodesic, while for $0<p<1$ and $p<0$, the longest weighted path $P^{(W)}_{\max}$ is closest. This is consistent with the fact that the continuum ordered geodesic is itself the corresponding optimal path in the continuum limit. The fact that the geometry of the DAG is encoded in the edge weights is not a big surprise. The weights are the ordered distances between points and so they carry a lot of the geometric information about the embedding space into the DAG. Unfortunately, when dealing with a DAG derived from data, there is no explicit embedding space so we do not normally have these weights available.

On the other hand, we found similar results for the optimal hop-based paths: the shortest-hop path tends to provide the better approximation in the metric regime, while the longest-hop path tends to do so in the reverse-metric regimes, as illustrated in \figref{fig:allpaths}. This was not guaranteed, since hop length contains no explicit information about the ordered distance. The same weighted-optimised path/hop-optimised path relationship is also recovered in the preferential-attachment citation model, where the geometry of the embedding space has been diluted further by additional processes when constructing the DAG.
This result is particularly useful for networks obtained from data, where there may be no intrinsic embedding space and no intrinsic notion of geometric distance, while hop length can always be defined directly from the graph. The correlation between hop length and weighted length suggests that, despite being blind to the latent space, hop length can still capture important features of its geometry. This broad distinction is observed across dimensions $D = 2, 3, 4$, although the hop-based signal is weaker and more variable than for the weighted paths. This is simply because the low integer values of hop length likely to be encountered are less able to distinguish between paths (see \citet{PE25} for a wider discussion of this issue).

\begin{figure}[H]
    \centering
    \includegraphics[width=\linewidth]{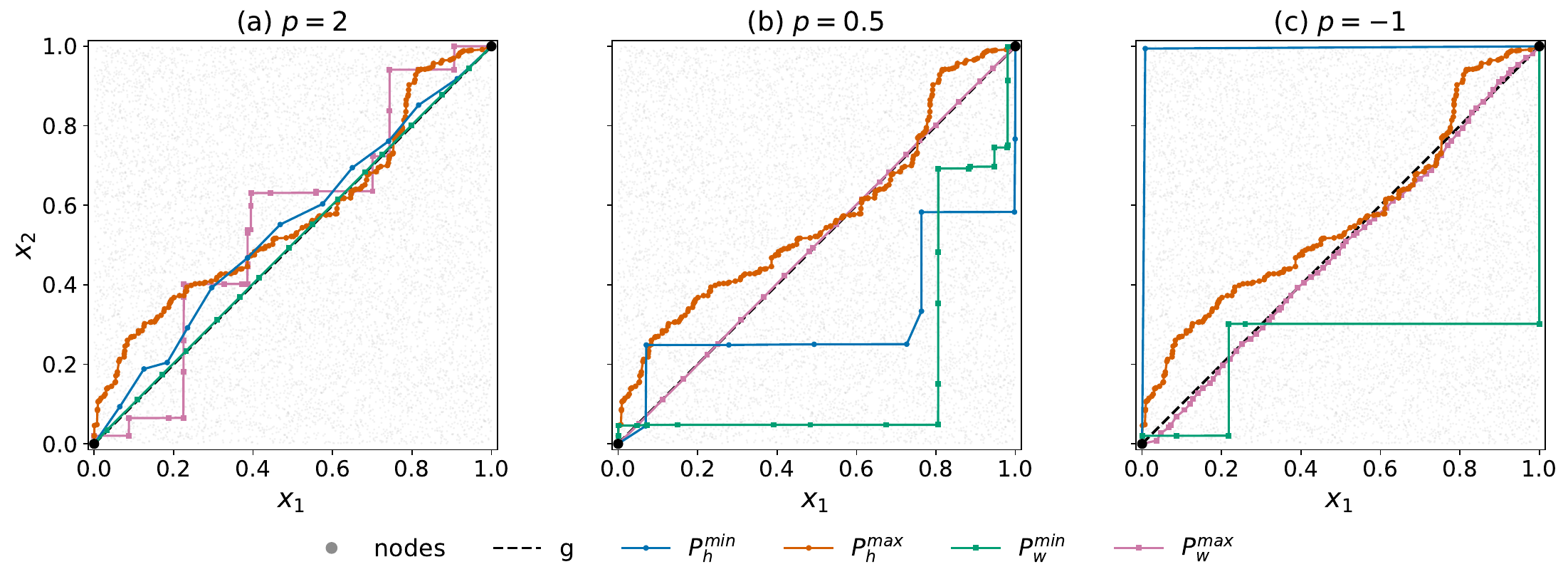}
    \caption{Source--sink paths in a $2D$ random geometric DAG across the different metric regimes. The same realisation of $\Ntilde=10{,}000$ nodes is used in all panels, with $\tseexpect{\kout}= 100$.}
    \label{fig:allpaths}
\end{figure}

Another result is our construction of an ordered space with three different geometric regimes. The different $p$ regimes of our \Lp{}-cube spacetime have different geometric features. The cases $p>1$ and $0<p<1$ have finite forward connection volumes, while still allowing us to choose between geometries with different ordered triangle identities.
The $p<0$ regime is novel but may at first appear redundant and of little practical use. However, its forward connection volume is infinite, making it closer to the situation in relativity, where the corresponding volume of the forward connection kernel is also infinite. There may therefore be situations in which this property is physically relevant and the $p<0$ regime is preferred. 

More generally, our results show that there is a choice of ordered geodesic in a DAG. When a researcher chooses to work with either the shortest or the longest path, what the researcher is effectively choosing is the geometry of the DAG, something which must be reflected in any latent space in which the DAG is assumed to live. For example if trying to embed a DAG in Minkowski spacetime, one should be using longest paths in the DAG \citep{CE17}. 
The longest path may be more informative when trying to follow innovation in citation networks \citep{HPEO24,HPEO24a,HPEO25,PE25}, since an innovation may develop through a sequence of small intermediate advances, each of which contributes to the final result. This emphasis on longer chains is also consistent with transitive reduction, which removes transitive shortcut edges while preserving reachability and thereby retains the longest source-target paths in a DAG \citep{CGLE14,AVE26}.
By contrast, a shortest path may be useful when the aim is to minimise the number of intermediate steps, for example when trying to trace the most accurate copies of a historical document or following a rumour back to its original source. For scientific information, this preference for shorter chains is also reflected in the concept of intermediacy, which tends to favour publications lying on shorter source--target paths \citep{SWTE20}. In this case one should look towards ordered spaces obeying ordered triangle identities if trying to embed the citation network.

There are several limitations to the current work. The interval DAG assumption, where every node lies on at least one path from source to sink, is a simplification: when analysing real citation networks, one may have several different key papers in mind which can act as multiple sources or sinks. The effect of such nodes on path observables is not addressed here. 

A natural direction for future work is to understand why the hop-based paths reproduce the same broad dependence on the underlying geometry as the weighted paths. One possibility is that the edges forming the optimal hop paths have characteristic weighted lengths. For the shortest-hop path, this scale may be set by the cutoff radius $R$, while for the longest-hop path it may be set by the typical inter-point scale, $a=N^{-1/d}.$
This would imply a strong relationship between hop length and weighted length for the optimal paths. Measuring the edge-length distributions directly would test this possibility.  Extending the framework to more general DAG structures, and to network types beyond citation networks such as causal graphs, are further directions for future work.

\subsection*{Data Accessibility} 
Copies of the \href{https://doi.org/10.6084/m9.figshare.34024839}{code} used in this work are available online \citep{AE26a}.

\subsection*{Funding}
AH is supported by a PhD scholarship from the Saudi Arabian Ministry of Education.

\begin{center}
	\LARGE\textbf{Appendices}
\end{center}
\appendix
\renewcommand{\thesection}{\Alph{section}}
\renewcommand{\theequation}{\thesection\arabic{equation}}
\renewcommand{\thefigure}{\thesection\arabic{figure}}
\renewcommand{\thetable}{\thesection\arabic{table}}
\numberwithin{equation}{section}
\numberwithin{figure}{section}
\numberwithin{table}{section}
\setcounter{section}{0}
\renewcommand{\theHsection}{\Alph{section}}
\renewcommand{\theHequation}{\thesection\arabic{equation}}
\renewcommand{\theHfigure}{\thesection\arabic{figure}}
\renewcommand{\theHtable}{\thesection\arabic{table}}

%
%
\input{daggeometrypathsappendix}

%


\newpage
\addcontentsline{toc}{chapter}{Bibliography}



\bibliographystyle{agsm}


\input{daggeometrypathsref}
\end{document}

%% file: daggeometrypathsappendix.tex

\section{Space, Time, Order and Distance}\label{s:space}

In this section, we summarise some of the formal definitions surrounding the concept of distance between points. Our main focus is on the case where there is a sense of order on our points, such as found in directed networks and in continuous spacetimes. This also helps to define our notation and terminology since the definitions we use are a hybrid of concepts from both discrete directed networks and from continuous spacetimes. 

In this section, we will refer most of the time to a set of points $\pointset$. This includes cases where this is an infinite set of points in continuous spacetime. However, the results also apply equally to the case where the points are a finite set of nodes in a network, usually denoted as $\Vcal$. 

Here will need our points $\Xcal$ to be connected by a set of \tsedef{directed edges} $\Ecal$. We will denote a directed edge from point $x$ to point $y$ as $\dedge{x}{y}$. Edges are an intrinsic part of the definition of a network\footnote{For undirected networks, if there is an undirected edge $\{u,v\}$, we assign two directed edges between two nodes $u$ and $v$ connected by an undirected edge. That is we have both $\dedge{u}{v}$ and $\dedge{v}{u}$ and work with this symmetrised directed graph.}. 
For continuous spacetimes, we define a directed edge $\dedge{x}{y}$ to exist when there is at least one curve from $x$ to $y$. In Euclidean space this is simple as all points are connected to all other points, so the set of directed edges $\Ecal$ is the set of all ordered pairs of points $\Ecal = \pointset \times \pointset$. In Minkowski spacetime, the set of edges is much more restricted as many points can be space-like separated, that is there is no path between them in either direction.  Here we define the directed edges in Minkowski space to be the time-like separated pairs of points in causal order, $\Ecal = \{(x,y) | \tau^2(x,y)>0 , x_0 < y_0\}$. Even when there is a curve from $x$ to $y$ the edge in the reverse direction is never present because of causality. 

Some of the following definitions are standard when working with networks with a finite set of nodes $\nodeset$ and a set of directed edges
$\Ecal$. Other concepts are more familiar in work on continuous spaces. For simplicity, we will just refer to points in the set $\pointset$ and a set of directed edges $\Ecal$ between pairs of points to cover both discrete networks and continuous spaces.

\subsection{Distances and Metrics}\label{as:distance}

In mathematics a \tsedef{distance} is a real function of two points in a set $\pointset$ that obeys the following axioms
\begin{definition}[Distance]\label{d:distance}
	A \tsedef{distance} on a set $\pointset $ is a function $d : \pointset  \times \pointset  \to \mathbb{R}$
	such that for all $x,y \in \pointset $:
	\begin{subequations}
		\begin{align}
			\text{Non-negativity:} & \;\; d(x,y) \geq 0 \, , \\
			\text{Reflexivity:}    & \;\; d(x,x) = 0 \, , \label{ae:reflexivity} \\
			\text{Symmetry:}       & \;\; d(x,y) = d(y,x) \, .
		\end{align}
		\label{ae:distancedef}
	\end{subequations}
\end{definition}
In fact for the rest of this discussion we will work with distance-like functions with a stronger condition, indiscernibility, replacing reflexivity\footnote{The only place where this matters in our discussion is when we consider null curves in relativity.}. That is we will invariably demand that $d(x,y) = 0$ if and only if $x = y$ so that $d(x,y)>0$ unless $x=y$.

If a distance function satisfies the \tsedef{triangle identity} then it is a \tsedef{metric}, that is 
\begin{definition}[Metric]\label{d:metric}
A \tsedef{metric} on $\pointset $ is a distance $d : \pointset  \times \pointset  \to \mathbb{R}$
that satisfies, for all $x,y,z \in \pointset $ :
\begin{subequations}
	\begin{align}
		\text{Non-negativity:}            & \;\; d(x,y) \ge 0 \, , \\
		\text{Indiscernibles:}            & \;\; d(x,y) = 0 \; \text{if and only if} \; x = y \, , \\
		\text{Symmetry:}                  & \;\; d(x,y) = d(y,x) \, , \\
		\text{Triangle inequality:}       & \;\; d(x,z) \leq d(x,y) + d(y,z) \, .
	\end{align}
	\label{ae:metricdef}
\end{subequations}	
\end{definition}
Note in particular that the triangle identity is true for any permutation of the same three points $x$, $y$ and $z$, so  $d(x,y) \leq d(x,z) + d(z,x)$ and $d(y,z) \leq d(y,x) + d(x,z)$ must also hold.

A \tsedef{metric space} is a space $\pointset$ with a metric $d$.

\subsubsection{Minkowski Distance}

One well known example of a distance function is the \tsedef{Minkowski distance}\footnote{Minkowski distance has no mathematical relation to the proper-time used to measure separation of points in Minkowski spacetime} which is one generalisation of the well known Euclidean distance. The Minkowski distance $d_p(x,y)$, also called the \Lp{} norm or simply the $p$-norm distance, between two points $x=(x_1,x_2,\dots,x_D) \in \Rbb^D$ and $y=(y_1,y_2,\dots,y_D) \in \Rbb^D$ in real space $\pointset \subseteq \Rbb^D$,  is
\begin{equation}
	d_p(x,y)  	
	=\Big(\sum_{i=1}^D |x_i-y_i|^p \Big)^{1/p} \, .
	\label{ae:lpdef}
\end{equation}
For $p \geq 1$, the Minkowski distance $d_p$ satisfies the triangle
inequality and hence defines a metric on $\mathbb{R}^D$ \citep{DD09}. 
Particular examples of the Minkowski distance include
\begin{itemize}
	\item the Euclidean distance for $p=2$,	$d_2(x,y) = (\sum_{i=1}^D (x_i-y_i)^2 )^{1/2}$,
	\item the Manhattan distance for $p=1$, $d_1(x,y) = \sum_{i=1}^D |x_i-y_i|$,
	\item the Chebyshev distance for $p\to\infty$, $d_\infty(x,y) = \max_i ( |x_i-y_i|)$.
\end{itemize}

\subsection{Paths, Geodesics and Generalised Distances}\label{as:pggd}

The main aim of this section is to give the definition of a geodesic that we use in our work. To do that we will carefully define what we mean by distance based on the length of a path such that this applies for both undirected and directed networks. In particular, we will drop the symmetry requirement of \defref{ae:distancedef} and define what we call a ``generalised distance''. We will consider the more specialised case of directed networks with no cycles, networks with partial orders, in the next section, \secref{as:osetdistancemetric}.

We start with the standard definition of a walk used in network analysis.
\begin{definition}[Walk] \label{d:walk}
	A \tsedef{walk} $\walkseq$ in a space $\pointset$ is a \emph{sequence of points} $x_i \in \pointset$, denoted as $\walkseq = (x_0,x_1,\ldots,x_\ell)$ where each point is connected to the next point in the sequence by an edge, so $\dedge{x_i}{x_{i+1}} \in \Ecal$. 
\end{definition}

We use the existence of a walk between two points in one direction to define a directed reachability relation.
\begin{definition}[Reachability] \label{d:reachability}
	A point $y$ is \tsedef{reachable} from point $x$ if there is a walk from $x$ to $y$. We denote this as $x\twoheadrightarrow y$. 
\end{definition}
As we want to work with directed networks, we do not assume this is symmetric, so $x\twoheadrightarrow y$ does not imply $y\twoheadrightarrow x$. In a Euclidean space, this ``reachability'' criteria is trivial as all points are connected to every other point. In a spacetime, such as Minkowski space, this is not trivial as reachability is defined in terms of causality and two points may not be connected or they are only connected in one direction. 

It is important for several of the results that follow that we include \tsedef{trivial walks} $\walkseq_x = (x)$ as a legitimate walk.  
\begin{definition}[Trivial walk] \label{t:walk}
	A \tsedef{trivial walk} $\walkseq_x$ is the sequence of a single point $x\in\pointset$, so $\walkseq_x=(x)$. 
\end{definition}

We can also join two walks to create a new walk if the second walk $\walkseq^{\prime\prime}$ starts where the first walk $\walkseq^{\prime}$ finished\footnote{This shows that our reachability relation $\twoheadrightarrow$ is transitive. It is reflexive $x\twoheadrightarrow x$ because of teh existence of trivial walks. However, this is not enough for the reachability relation $\twoheadrightarrow$ to be a partial order.}. More formally:-
\begin{definition}[Walk concatenation] \label{d:walkconcat} 
	The concatenation of two walks $\walkseq^{\prime} = (x_n)_{n=0}^{N}$ and  $\walkseq^{\prime\prime} = (y_m)_{m=0}^{M}$ to create a new walk $\walkseq = (z_r)_{r=0}^{M+N} $ denoted as $\walkseq = \walkseq^{\prime} \circ \walkseq^{\prime\prime}$ is defined to be 
	\begin{subequations}
		\begin{eqnarray}
			\walkseq = (z_r)_{r=0}^{M+N} 
			&=&
			(x_n)_{n=0}^{N} \circ (y_m)_{m=0}^{M} 
			\\
			&=&
			(x_0, x_1, \ldots, x_{N-1}, x_{N} = y_0 , y_{1}, \ldots, y_{M-1}, y_{M} )
			\\
			\quad \text{where}
			\quad
			z_r
			&=&
			\begin{cases}
				x_r     & \text{if} \;\; 0 \leq r \leq N , \\
				y_{r-N} & \text{if} \;\; N \leq r \leq M+N .
			\end{cases}
		\end{eqnarray}
		\label{e:walkconcat}
	\end{subequations}
\end{definition}
Note that a walk is simply a concatenation of the sequence of edges that can also be used to define any walk, $E(\walkseq) = ( \dedge{x_{i-1}}{x_i})_{i=1}^{i=\ell}$.

Walks are important in networks as they define which nodes are connected. Our generalisation of edges to the continuous space case give the same interpretation in that case, the existence of a curve from point $x$ to point $y$ means $x$ is connected to $y$. We may interpret the curve as a walk from $x$ to $y$, one made up of the points on the curve. What is important is that either way we have a set $\Ccal$ of pairs of points which are connected, directly by edges or indirectly by walks. 

A \tsedef{path} $P$ is a walk where no two nodes are the same.
\begin{definition}[Path]\label{d:path}
	A \tsedef{path} $P$ in a space $\pointset$ is a \emph{sequence of points} $x_i \in \pointset$, denoted as $P = (x_0,x_1,\ldots,x_\ell)$ where each point is connected to the previous point and it is distinct, so
	$x_i=x_j$ where $x_i,x_j \in P$ if and only if $i=j$. 
\end{definition}
All paths are walks but the reverse is not usually true in a generic network because of cycles. However, for a directed acyclic graph (DAG), all walks are also paths because of the lack of cycles.  	
The trivial walks are also paths, trivial paths, $P_x = (x)$.

A notation we will also use is that $\pathset$ is the set of all possible paths in the space $\Vcal$ while $\pathset(x,y)$ is the \emph{set} of \emph{all} paths which start an initial point $x\in\pointset$ and end at $y \in\pointset$. 
Sometimes we will use $P_{xy}$ to denote \emph{one} such path running from $x$ to $y$, so that  $P_{xy} \in \pathset(x,y)$.

We also need a function $L$ that gives the \tsedef{length} $L(\walkseq)$ of a walk $\walkseq$. For our work, a length function must obey the following axioms\footnote{It is quite feasible to replace the indiscernibles axiom with the reflexivity axiom of \eqref{ae:reflexivity} (zero for trivial paths but can be zero for other paths). The indiscernibles axiom seems more relevant in the network context while the reflexivity condition \eqref{ae:reflexivity} will allow for null paths in spacetimes.}
\begin{definition}[Length of a walk] \label{d:pathlength}
	The \tsedef{length} $L(\walkseq)\in \Rbb$ of a walk $\walkseq$ is a function that obeys the following axioms\footnote{It is possible to generalise this further. We are using the natural concept of addition for real numbers but we can use other algebraic structures for example Group Graphs \citep{E25}. Some of our later results will not apply for such generalisations and we do not consider them here.}
	\begin{subequations}
		\begin{align}
			\text{Reality:}        \quad & L(\walkseq) \in \Rbb  , \\
			\text{Non negativity:} \quad & L(\walkseq) \geq 0  , \\
			\text{Indiscernibles:} \quad & L(\walkseq)   =    0   \;\; \text{iff} \;\; \walkseq=(x), 
			\label{d:lengthindisc}
			\\
			\text{Additivity:}     \quad & L(\walkseq^{\prime} \circ \walkseq^{\prime\prime})   =    L(\walkseq^{\prime}) + L(\walkseq^{\prime\prime})  .
			\label{ae:lengthadd}
		\end{align}
		\label{ae:lengthdef}
	\end{subequations}
\end{definition}
This is for any walk $\walkseq$ in the space $\pointset$ including all trivial paths $(x)$. 

Since a single directed edge is always a walk, it follows from the non negativity and indiscernibles properties of length that:
\begin{lemma}[]\label{l:lengthofedge}
	The length of every edge (even those forming self-loops) must be a positive number.
\end{lemma}
This includes the length of edges in self-loops. 
Such edge lengths are often defined to be equal to the edge weights so in that case our length axioms in \defref{d:pathlength} put limitations on what edge weights can be used as edge lengths. The indiscernibles and non negativity axioms imply that only the trivial paths have zero length\footnote{We are therefore excluding a case where we want to include null paths, paths of length zero. These correspond to paths followed by light in relativity contexts. Such null paths have zero proper-time lengths and they play an important role in relativity. To work with such null paths we would have to replace the indiscernibles axiom for length by reflexivity such as used in \eqref{ae:odreflexivity}. This is another possible route but not one we will follow here.}.

Since any walk $\walkseq$ can be defined in terms of the sequence of edges in that walk, $E(\walkseq)$, the additivity axiom of walks \eqref{ae:lengthadd} means that:-   
\begin{lemma}[]\label{l:lengthsumofedges}
	The length $L(\walkseq)$ of a walk $\walkseq = (x_i)_{i=0}^\ell$ is the sum of the length of the edges $E(\walkseq)$ in the walk.
	\begin{equation}
		L(\walkseq) 
		= \sum_{e \in E(\walkseq)} L( e ) 
		= \sum_{i=1}^\ell L\big( (x_{i-1},x_{i}) \big) \, .
		\label{ae:lengthsumofedges}
	\end{equation}
\end{lemma}

The indiscernibles axiom for lengths \eqref{d:lengthindisc} and \lemref{l:lengthsumofedges} together imply the following.
\begin{corollary}\label{c:lengthwalkinc}
	The length of a walk is always a monotonically strictly increasing function of the number of nodes in the walk.
\end{corollary}
This seems like a sensible property for the length of a walk to have and so this gives support to the system of definitions given here\footnote{The one possible exception would be null curves in continuous spacetimes (as followed by light in a vacuum) which have zero ``length'' i.e.\ proper time. So as it stands our approach does not include these cases and they do not appear in any of our discrete spacetime models.}.

Let us look at some examples. If our space $\pointset$ is the set of nodes of a graph, the definitions of walk and paths coincide with the definitions used in graph theory and network science. 
In particular, if we use the standard unit weights assigned to edges in a network as the length of each edge, $L(e)=1 \; \forall \; e \in \Ecal$, then the length of any walk is equal to the number of points in the walk minus one. Again that is the default definition of path length in network science \citep[e.g.][sec.7, p.111]{C21}. 

The length of a walk in this case can be the number of edges in the walk, that is one less than the number of points. For a continuous space, such as Euclidean space, the Euclidean distance between two points acts is a suitable length function but so can any of the $L_p$ measures $d_p(x,y)$ of \eqref{ae:lpdef}. 

We can now define the \tsedef{geodesic} $\gamma_L(x,y)$ running between two points $x$ and $y$ in our space $\pointset$ given a path length function $L$.
Qualitatively, the geodesic is the ``path of least resistance'', the most ``direct'' way of getting from $x$ to $y$. Formally, we define the geodesic to be a \emph{path}\footnote{When dealing with minimal geodesics, shortest paths, we could also relax this and consider all walks provided only the trivial walks have zero length (indiscernible axiom for length).} with an extremal length  between two points. This gives two types of geodesic, \tsedef{maximal geodesics} $\gammamax_L(x,y)$ using the longest paths and \tsedef{minimal geodesics} $\gammamin_L(x,y)$ using the shortest paths.
\begin{definition}[Geodesic]\label{d:geodesic}
	A \tsedef{geodesic} $\gamma_L(x,y)$ in a space $\pointset$ with a length function $L: \Pcal \to \Rbb$ defined by \defref{d:pathlength} is a \emph{path} from $x \in \pointset$ to $y \in \pointset$ that has an extremal length between those two points. This gives two types of geodesic, \tsedef{maximal geodesics} $\gammamax_L(x,y)$ and \tsedef{minimal geodesics} $\gammamin_L(x,y)$ where
	\begin{subequations}\label{ae:geodef}
		\begin{align}
			\gammamin_L(x,y) 
			&=
			\big\{
			\gamma_L \; | \;\gamma_L \in \pathset(x,y) , \;\; L(\gamma_L) \leq L(P) \; \forall \; P \in \pathset(x,y) 
			\big\}
			\;\; \text{(shortest path)},
			\label{ae:geosp}
			\\
			\text{or} &
			\nonumber \\
			\gammamax_L(x,y)
			&=
			\big\{
			\gamma_L \; | \; \gamma_L \in \pathset(x,y) , \;\; L(\gamma_L) \geq L(P) \; \forall \; P \in \pathset(x,y) 
			\big\}
			 \;\; \text{(longest path)}.
			\label{ae:geolp}
		\end{align}	
	\end{subequations}	
\end{definition}
Where $\pathset(x,y)$ is the set of all paths from $x$ to $y$. 
It is important that only one of these criteria is used for all pairs of points, we do not define our geodesics using the longest path for some pairs of points while the shortest path is used for the remaining pairs.

Some points to note. 
\begin{itemize}
	\item These definitions of geodesics work for any set $\pointset$, discrete or not, provided this comes equipped with a measure of length $L: \Pcal \to \Rbb$ of \eqref{ae:lengthdef} that is defined for any path between two points in $\pointset$.
	\item We can have completely different geodesics $\gamma_L$ on one network or space depending on the length function $L$ we use.
	\item There can be multiple geodesics between the same pair of points provided the geodesics have the same length. 
	\item There is no assumption that paths are reversible. That is if we have a geodesic $\gamma_L(x,y)$ from $x$ to $y$, this does not imply the existence of a geodesic from $y$ to $x$ and even if one does exist, $\gamma_L(y,x)$ need not be the same path reversed or even the same length as $\gamma_L(x,y)$. 
	\item We always have the trivial path from $x$ to itself which is just the point $x$. These are the only paths with one point, and so, regardless of the extremal condition used to define a geodesic, such trivial paths are the only geodesics of length zero $\gamma_L(x,x)$. This is property satisfies the indiscernible axiom \eqref{d:lengthindisc} for lengths. 
\end{itemize}

The fact that the length of a geodesic from any one point $x$ to any other point $y$ means this can be used as the basis for a distance-like function as follows.
\begin{definition}
	The \tsedef{geodesic distance} $d_L(x,y)$ from point $x$ to point $y$ is defined to be the length $L$ of the geodesic $\gamma_L(x,y)$ from point $x$ to point $y$ provided a path exists from $x$ to $y$
	\begin{equation}
		d_L(x,y) = L(\gamma_L(x,y)) \quad \text{if} \;\; x \twoheadrightarrow y \, .	
		\label{ae:geodistdef}
	\end{equation}
\end{definition}
Clearly we have to specify the length function being used and the maximal/minimal property of the geodesics used to define this $d_L$ function. The uniqueness of the length of a geodesic when it exists is important to ensure this is a function in the formal sense. We will specify the properties of such geodesic distances in our general directed below. To do this we first need to consider some of the properties of geodesics in this directed context.

We start by showing that we can split a geodesic into a concatenation of geodesics.\footnote{The corollary of this lemma is also true. That is if the concatenation of two paths gives a geodesic then the paths are also geodesics.} 
\begin{lemma}\label{l:geosplit}
	Every geodesic can be split at any point into two geodesics. So for any $z \in \gamma_L(x,y)$ where $\gamma_L(x,y) = P_{xz} \circ P_{zy}$ then $P_{xz}$ is a geodesic from $x$ to $z$ and $P_{zy}$ is a geodesic from $z$ to $y$. 
\end{lemma}
\begin{proof}
	Suppose we have split the geodesic at a point $z$ to give $\gamma_L(x,y) = P^\prime \circ P^{\prime\prime}$ where the path $P^\prime$ starts at $x$ and ends at $z$ while path $P^{\prime\prime}$ starts at $z$ and ends at $y$. Since the geodesic is a path with no repeated points, the subsequences $P^\prime$ and $P^{\prime\prime}$ are also paths with a sequences of distinct points. 
	
	Consider another path $P^{\prime\prime\prime}$ from $x$ to $z$ and suppose $P^{\prime\prime\prime}$ had a more extreme length than  $P^{\prime}$, 
	so either $L(P^{\prime\prime\prime}) < L(P^{\prime\prime\prime})$ if we use shortest paths \eqref{ae:geosp}
or $L(P^{\prime\prime\prime}) > L(P^{\prime\prime\prime})$ if we use longest paths \eqref{ae:geolp}. This means that $P^{\prime\prime\prime}$ is possible geodesic between $x$ and $z$ and certainly $P^{\prime}$ is not a geodesic. However, there is no guarantee that $P^{\prime\prime\prime}$ does not share many points with $P^{\prime\prime}$ in addition to $z$, $|P^{\prime\prime\prime} \cap P^{\prime\prime}|>1$. Then we can not combine $P^{\prime\prime\prime}$ and $P^{\prime\prime}$ to produce a path needed for potential geodesics.


	However, the properties of the lengths means $L(P^{\prime\prime\prime} \circ P^{\prime\prime})=L(P^{\prime\prime\prime}) + L(P^{\prime\prime})$ and based on our criteria this must be a more extreme length than the origional path. That then  contradicts the initial assumption that $\gamma_L(x,y)$ is a geodesic with the most extreme value of length possible. So there can be no path $L(P^{\prime\prime\prime})$ from $x$ to $z$ which has a more extreme length than the original $L(P^{\prime})$. This shows that $L(P^{\prime})$ is a geodesic\footnote{Note we do not exclude the possibility that $L(P^{\prime})=L(P^{\prime\prime\prime})$ in which case $P^{\prime\prime\prime}$ is also a geodesic from $x$ to $z$.} from $x$ to $z$. 
	
	A similar argument applies to the later part, $P^{\prime\prime}$, that tells this also has the best possible length. So we conclude that we can split any geodesic into two sub-sequences, sub-paths, and these subpaths are guaranteed to be geodesics. That is for any geodesic $\gamma_L(x,y)$ we have that 
	\begin{equation}
		\gamma_L(x,y) = \gamma_L(x,z) \circ \gamma_L(z,y) \quad \forall \; z \in \gamma_L(x,y) \, .
	\end{equation}
\end{proof}

We can also deduce the following property:
\begin{lemma}\label{l:geo2}
	If we find a point $z$ that lies on some path $P_{xy}$ from $x$ to $y$, so that $P_{xy} = P^{\prime}_{xz} \circ P^{\prime\prime}_{zy}$, then $d_L(x,y)   =    d_L(x,z) + d_L(z,y)$ if and only if $z$ lies on a geodesic from $x$ to $y$.
\end{lemma}
Again, this is simply follows from the fact that geodesics are paths with extremal values for their lengths. 
\begin{proof}
	Suppose $d_L(x,y)   =    d_L(x,z) + d_L(z,y)$. Since we have a path from $x$ to $z$ that then continues onto $y$ we have two geodesics $\gamma_L(x,z)$ from $x$ to $z$ and  $\gamma_L(z,y)$ from $z$ to $y$. If we concatenate these geodesics $\gamma_L(x,z) \circ \gamma_L(z,y)$ then the length of this combined path is equal to the extremal length of any geodesics from $x$ to $z$. That is, this combined path $\gamma_L(x,z) \circ \gamma_L(z,y)$ is a geodesic from $x$ to $y$ that passes through point $z$ as required.
	
	The reverse case states that if $z$ lies on a geodesic from $x$ to $y$, say $\gamma_L(x,y)$, then $d_L(x,y)   =    d_L(x,z) + d_L(z,y)$. If $z$ lies on a geodesic then we can use \lemref{l:geosplit} to show that we can spit the original geodesic $\gamma_L(x,y)$ into two parts, $\gamma_L(x,z)$ and $\gamma_L(y,z)$, each geodesics themselves. That is $\gamma_L(x,y) = \gamma_L(x,z) \circ \gamma_L(y,z)$. The additivity of path lengths over concatenation then gives us that $d_L(x,y)   =    d_L(x,z) + d_L(z,y)$ as required.
\end{proof}
Note that the path here $P_{xy}$ need not itself be a geodesic. All it does is ensure that there is connectivity from $x$ to $y$ via $z$. 
Also, in the proof of \lemref{l:geo2}, if $z$ is not on a geodesic then we must have strict inequalities from the triangle identities.
\begin{corollary}
	Suppose we find a point $z$ that lies on some path $P_{xy}$ from $x$ to $y$ which is \emph{not} a geodesic. Then we know that $d_L(x,y)  <   d_L(x,z) + d_L(z,y)$ if $\gamma_L(x,y)$ is a minimal geodesic or $d_L(x,y)  >   d_L(x,z) + d_L(z,y)$ if $\gamma_L(x,y)$ is a maximal geodesic.
\end{corollary}

The last step in this section is to look at the properties of the length of a geodesic $d_L(x,y)$. However, for our very general systems of undirected and directed networks, continuous space and continuous spacetimes, we need a weaker definition of distances as compared to the formal definitions of distance in \defref{d:distance} and metric in \defref{d:metric}.  So we will first define a \tsedef{generalised distance} $d$ before showing that the length of any geodesic\footnote{From the definition of a geodesic, all geodesics between two given points have the same length.} between two points gives satisfies this weaker definition of ``distance''. 

A \tsedef{generalised distance} is a map $d$ from any two connected points to real values that satisfies the following axioms.
\begin{definition}\label{d:gdistance}
	A \tsedef{generalised distance} $d$ from point $x$ to a point $y$ that is connected to $x$, $x \twoheadrightarrow y$,  is a function $d: \Ccal \mapsto \Rbb$ which obeys
	\begin{subequations}
		\begin{align}
			\text{Reality:} \quad & d(x,y) \in \Rbb  , \\
			\text{Non-negativity:} \quad & d(x,y) \geq 0  , \\
			\text{Indiscernibles:} \quad & d(x,y)   =    0   \;\; \text{iff} \;\; x=y .
		\end{align}
		\label{ae:gdistancedef}
	\end{subequations}
\end{definition}
Formally, the set of connected pairs is $\Ccal = \{ (x,y) | x \twoheadrightarrow y\} \subseteq \pointset \times \pointset$.
Note that the value of a generalised distance from $x$ to $y$ is not defined if the points are not connected in that direction\footnote{It maybe useful to extend the definition in some circumstances to give values when $x$ and $y$ are not connected but we do not use that option here.}, if there is no path from $x$ to $y$.  Also, a generalised distance need not be symmetric; it may be defined in one direction but not the other or it may have two different values so that $d(x,y) \neq d(y,x)$.

In the same way, we can add some sort of triangle identity to those of a generalised distance in \defref{d:gdistance} to get what we call a \tsedef{generalised metric}, a map from connected pairs of points to real numbers that satisfies the following axioms 
\begin{definition}\label{d:gmetric}
	A \tsedef{generalised metric} $d$ is a function of a pair of connected points, $d: \Ccal \mapsto \Rbb$, which obeys the following axioms
	\begin{subequations}\label{ae:gmetric}
		\begin{align}
			\text{Non-negativity:} \quad & d(x,y) \geq 0  , \\ 
			\text{Indiscernibles:} \quad & d(x,y)   =    0   \;\; \text{iff} \;\; x=y, \\ 
			\text{Triangle Identity:} \quad & d(x,y)   \leq d(x,z)   + d(z,y)    \;\; \forall \; x \twoheadrightarrow z \twoheadrightarrow y, 
			\label{ae:triiddef} \\
			\text{or} \nonumber \\
			\text{Reverse Triangle Identity:} \quad & d(x,y) \geq d(x,z)   + d(z,y)    \;\; \forall \; x \twoheadrightarrow z \twoheadrightarrow y.
			\label{ae:revtriiddef}
		\end{align}
	\end{subequations}
\end{definition}
Again, for any given generalised metric the function either satisfies the triangle identity \eqref{ae:triiddef} for every pair of connected points, or it always satisfies the reverse triangle identity \eqref{ae:revtriiddef} for every pair of connected points. 
Note we do not necessarily have symmetry so this is not a distance in the standard sense of \defref{d:distance}. 
So even if $d(x,y)$ has a value, $d(y,x)$ may not be defined or it can have a different value.
The additivity property of path length functions \eqref{ae:lengthdef} is also not carried over but is replaced by the triangle inequalities \eqref{ae:triiddef} or \eqref{ae:revtriiddef} as we shall show.

\begin{lemma}
	The length $L(\gamma_L(x,y))$ of any geodesic $\gamma_L(x,y)$ is a \tsedef{generalised metric} $d_L(x,y) = L(\gamma_L(x,y))$.
\end{lemma}
Note that the length function $L$ must be the same path-length function used to define the geodesics $\gamma_L$ in \defref{d:geodesic}.
\begin{proof}
	First we note that geodesics exist if and only there is a path from $x$ to $y$. So the length of a geodesic has the correct domain and codomain needed for a generalised distance function. 
	
	Next, the axioms of a length function in \defref{d:pathlength} include all the properties required by a generalised distance as defined in \defref{d:gdistance}. 
	
	Finally, the extremal properties of the lengths of geodesics that are the core of the definition of a geodesic in \defref{d:geodesic} ensure that they obey a triangle or reverse triangle inequality as noted in \lemref{l:geo2}.
\end{proof}

\subsection{Order, sets, distance and metrics}\label{as:osetdistancemetric}

Our work is with directed acyclic graphs, so we need to generalise the concepts of distance and metric further to take account of situations where there is an order inherent in the network. As before, we can extend these ideas to continuous spaces by defining an edge to exist whenever there is a curve from one point to another in a continuous space. Minkowski spacetime then provides the best known reference point and test case for these concepts.  First, we start with a reminder of how order in a set is described.

\subsubsection{Generic Partial Orders}\label{as:partialorder}

A \href{https://en.wikipedia.org/wiki/Partially_ordered_set#Formal_definition}{\tsedef{partial order}} $\preceq$ on a set $\pointset$ (sometimes called the \tsedef{ground set} of the poset) is a binary relation which satisfies three conditions:
\begin{subequations}
	\begin{align}
		\text{reflexive:}    \; & x \preceq x \, , 
		\label{e:poreflexivity}
		\\
		\text{antisymmetry:} \; & x \preceq y \; \text{and} \; y \preceq x \; \text{iff} \; x=y \,, 
		\label{e:poantisymmetry}
		\\
		\text{transitive:}   \; & \text{if} \; x \preceq y \; \text{and} \; y\preceq z \; \text{then} \; x \preceq z \, .
		\label{e:potransitivity}
	\end{align}
	\label{ae:porderdef}
\end{subequations}
where these hold for any $x,y,z \in \pointset$. The ground set $\pointset$ and the binary relations $\preceq$ obeying \eqref{ae:porderdef}  together form a \tsedef{partially ordered set} or \tsedef{poset} which we can denote as $\Pcal = (\pointset,\preceq)$. 
Note that $y \succeq x$ is an alternative notation for $x \preceq y$. 

It will be useful sometimes to define a related relationship between two objects in the set, $\prec$, which is the same as the binary relation of a partial order $\preceq$ except that we exclude the reflexive case where $x = y$.  That is 
\beq 
x \prec y \quad \text{if} \quad x \preceq y  \;\; \text{and} \;\; x \neq y \, .
\eeq 
Again, $y \succ x$ means the same as $x \prec y$. 

\subsubsection{Partial Orders in this Paper}\label{as:partialorderhere}

In this paper we have two partial orders. We use the symbol $\trianglelefteq$ for the partial order defined on points in $\Rbb^D$ using the cube spacetime rule of \eqref{ae:cubespacerule} and we use the symbol $\preceq$ for the partial order on nodes of a DAG defined by the existence of a path from one node to another. That is 
\begin{equation}
	\forall \; u,v \in \Vcal \quad 
	\begin{cases}
		u     \preceq v & \Leftrightarrow \quad  \text{a path from node }u\text{ to node }v\text{ exists,} \\
		u \not\preceq v & \Leftrightarrow \quad \text{no path from node }u\text{ to node }v\text{ exists.}
	\end{cases}
	\label{ae:preceqdef}
\end{equation}
In \secref{as:pggd} we used $x\twoheadrightarrow y$ to indicate that a path exists from $x$ to $y$ so both our binary relations in this section, $x \trianglelefteq y$ and $x \preceq y$, also imply  $x\twoheadrightarrow y$. We have changed the notation to emphasise the fact that reachability when there are no cycles has additional implications. In particular, one can show that the cubespace rule $\trianglelefteq$ of \eqref{ae:cubespacerule} and node reachability in a DAG $\preceq$ always satisfy the criteria \eqref{ae:porderdef} for a partial order relation whereas in the previous section our reachability relation $\twoheadrightarrow$ need not.

\subsubsection{Ordered distance and metrics}\label{as:odistancemetric}

In our work with directed acyclic graphs, we need to define weaker versions of the standard mathematical concepts of distance and metric that can work with spaces with order, what we call spacetimes. These concepts will also be useful where working with continuous spacetimes such as Minkowski spacetime and de Sitter spacetime.

An \tsedef{ordered distance} $d_\mathrm{o}(x,y)$ between any two points $x,y\in \pointset$ in a partially order set $(\pointset,\trianglelefteq)$ satisfies 
\begin{subequations}
	\begin{align}
		\text{connectivity}   & \;\; d_\mathrm{o}(x,y) \; \text{is defined only if} \; 
		x \trianglelefteq y \, , 
		\label{e:connectivity} 
		\\
		\text{non-negativity} & \;\; d_\mathrm{o}(x,y) \geq 0 \;  \text{if} \;  x \trianglelefteq y \, ,
		\\
		\text{reflexivity}    & \;\; d_\mathrm{o}(x,x)=0 \, .
		\label{ae:odreflexivity}
	\end{align}
	\label{ae:odistancedef}
\end{subequations}
Note that for any pair of points with no order, that is $x \not\trianglelefteq y$, then an ordered distance $d_\mathrm{o}(x,y)$ is not defined\footnote{A different approach is used for Lorentzian length spaces as defined by \citet{KS18}. There, in Definition 2.8, they work with ``time separation function'' which plays the role of our ordered distance $d_\mathrm{o}$.  A time separation function is also defined in terms of operators providing an order such that its is real and positive when $x \triangleleft y$ (in our notation) but set to zero when $x \not\!\triangleleft\, y$. Note that this second case contains both $x=y$ and $x \not\trianglelefteq y$ cases. So the difference with our definition is for the unconnected pairs of points (space-like separated), we leave that case as ``undefined'' while it is zero for a time separation function. It seems likely that assigning a real value for pairs of points with no direct connection is a matter of convenience since the lack of direct connection means one might not expect there to be a natural distance. We prefer to keep the value for such connected pairs to be undefined so as to be as general as possible.}. This means that for any pair of points either $d_\mathrm{o}(x,y)$ or $d_\mathrm{o}(y,x)$ is not defined or both are not defined. 

We can have two types of ordered metrics. An \tsedef{ordered metric} is an ordered distance that satisfies a \tsedef{ordered triangle identity}
\beq
	d_\mathrm{o}(x,z) \leq d_\mathrm{o}(x,y) + d_\mathrm{o}(y,z) 
	\, ,
	\quad  \text{where} \quad 
	x \trianglelefteq y \trianglelefteq z \, .
	\label{ae:omtridef}
\eeq
This is weaker than the usual triangle identity for a normal metric since in this case the inequality only needs to hold for one permutation of the points $x$, $y$ and $z$. For instance we do not need $d_\mathrm{o}(x,z) + d_\mathrm{o}(y,z)$ to have any relation to $d_\mathrm{o}(x,y)$ .

An \tsedef{ordered reverse metric} is an ordered distance that satisfies the \tsedef{reverse ordered triangle identity}
\beq
	d_\mathrm{o}(x,z) \geq d_\mathrm{o}(x,y) + d_\mathrm{o}(y,z) 
	\, ,
	\quad  \text{where} \quad 
	x  \trianglelefteq y \trianglelefteq z \, .
	\label{ae:ormtridef}
\eeq
Again, note the order of the points is important.

Finally, whenever we have an ordered metric, we can always define a \tsedef{ordered geodesic} $\gamma_\mathrm{o}(x,y)$ from point $x$ to point $z$ to be a set of points $z$ which satisfy the equality of the relevant ordered triangle inequality \eqref{ae:omtridef} or \eqref{ae:ormtridef}, that is
\begin{equation}
	\gamma_\mathrm{o}(x,z) = \{ y \; | \; d_\mathrm{o}(x,z) = d_\mathrm{o}(x,y) + d_\mathrm{o}(y,z) \} 
	\label{e:omgeodesic}
	\, .
\end{equation}
Following this nomenclature, we should refer to minimal ordered geodesics and maximal ordered geodesics depending on the type of ordered triangle identity satisfied by the geodesics.

\subsubsection{Proper Time}\label{as:propertime}

In relativity, a common notation and choice of coordinates for a point in spacetime is $x_\mu = (x_0, x_1, \ldots, x_{D-1})$ where the zero-th coordinate $x_0$ is the time and the remaining $(D-1)$ coordinates are the spatial coordinates. 
A useful measure of the separation of points $x$ and $y$ is given by the proper-time $\tau(x,y)$ provided there is a causal connection between the two point\footnote{This defines a partial-order where $x \preceq  y$ indicates that light or some massive particle can travel from $x$ to $y$. In our work, the use of the indiscernible axiom in our definitions means we would exclude the null paths, paths of zero proper-time length, when defining the binary relation of our poset.} The general form for the proper-time along a path is found by integrating the infinitesimal form
\bea
 	d\tau^2 
 	&=& 
 	\sum_{\mu=0}^{D-1}\sum_{\nu=0}^{D-1} g^{\mu\nu} dx_\mu dx_\nu \, .
 	\label{ae:dtaudef}
\eea
Here the matrix $g^{\mu\nu}$ is the \tsedef{metric tensor}\footnote{Confusingly, in much of the literature on relativity, this name is often shorten to just the ``metric''. However $g^{\mu\nu}$ is \emph{not} a metric in the sense defined in \eqref{ae:metricdef} and so we will call $g^{\mu\nu}$ the ``metric tensor''.} and this encodes information about the curvature of the spacetime. One of the most important properties of proper time is that it is invariant under the symmetries of the space. That is all observers agree on the same value of proper time for two points (spacetime events) $x$ and $y$ and this means they all agree that the physical outcomes are the same. 

The simplest example is Minkowski spacetime, which has no curvature but maximal symmetry (it is homogeneous and isotropic, so distances between pairs of points are unchanged by rotations or translations) in which case we can choose coordinates and conventions such that
\beq
	g^{\mu\nu}
	=
	\begin{pmatrix}
		+1 &  0 &  0 & 0 \\
		0  & -1 &  0 & 0 \\
		0  &  0 & -1 & 0 \\
		0  &  0 &  0 & -1 
	\end{pmatrix}.
\eeq
Then the proper time separation $\tau(x,y)$ between two points $x$ and $y$ in Minkowski spacetime is given by a formula for the square $\tau^2$. For standard coordinates this is
\beq
	\tau^2(x,y) 
	= \big( (x-y)_0 \big)^2 - \sum_{i=1}^{D-1} \big((x-y)_i\big)^2 \, .
	\label{e:propertimemst}
\eeq
Note that the minus sign on the spatial components $(x-y)_i$ for $i=1,2,\ldots,(D-1)$ leads to the result that means that $\tau^2(x,y)$ is only non-negative when no physical particle can travel between the two points. Further, by convention, when points are connected we always define the proper time difference to be the positive square root so $\tau(x,y)\geq 0$.  Otherwise this formula can give negative and imaginary values for $\tau$. So clearly, the standard definition of distance \eqref{ae:distancedef} with its non-negative values cannot apply here.

While proper time may not be a distance \eqref{ae:distancedef} or a metric  \eqref{ae:metricdef}, we can use proper time $\tau(x,y) \geq 0$ when there is a path for a massive free particle to move from $x$ to $y$, when they are causally connected to define what we call here an ordered distance $d_L(x,y)$ satisfying \eqref{ae:odistancedef} for these relativistic (Lorentzian) spaces where\footnote{Our use of the indiscernibles condition means we are excluding null paths here, paths followed by massless particles.} 
\begin{equation}
	d_\mathrm{O}(x,y)
	= 
	\begin{cases}
		       \tau(x,y) & \text{if} \; \tau^2(x,y) >0 \;\; \text{and} \;\; x_0 < y_0 , \\
		               0 & \text{if} \; x =  y , \\
		\text{undefined} & \text{otherwise}.
	\end{cases}
	\label{ae:ldistancedef}
\end{equation}

\section{The \Lp{}-cube spacetime model}

We work with Minkowski distance in the context of our \Lp{}-cube spacetime models where there is a cutoff $R$ in the length of any edge in terms of the distance as measured using Minkowski distance. Since we prefer to compare results for networks with the same average degree, similar density of edges, we need an expression for the degree in terms of this cutoff $R$. So the aim here is to find $\Omega_{p,D}(R)= |\Xcal^+(x,R)|$ the volume of the forward connection kernel $\Xcal^+(x,R)$ of any point $x$ in $D$-dimensional \Lp{}-cube spacetime. 

The forward connection kernels $\Xcal^+(x,R) = \{ y | d_p(x,y) \leq R \}$ for the various $p$ values are shown in \figref{f:dp-ball}.
\begin{figure}[H]
	\centering
	\includegraphics[width=\linewidth]{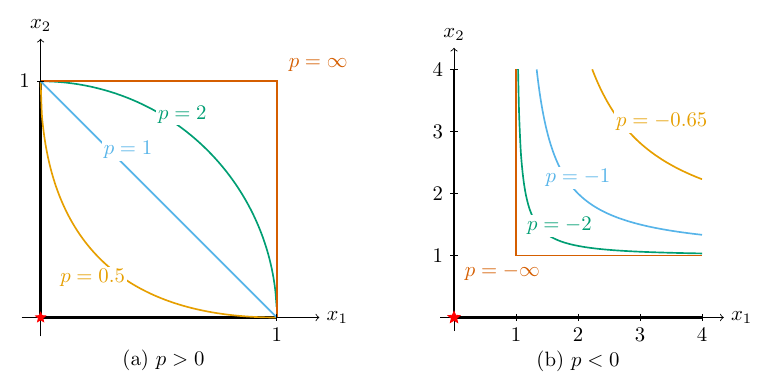} 
	\caption{For the \Lp{}-cube model in two dimensions, the node at the origin (red star) can only connect to nodes in the forward connection kernel with unit cutoff, $\Xcal^+(x=0,R=1)= \{ y | d_p(y,0) \leq 1 \}$. One boundary of the forward connection kernel $\Xcal^+(x=0,R=1)$ is given by the edge of the unit balls $d_p(0,x)=1$ and these are shown in colour for various $p$. The thick black lines show the other two sides of the boundary of the forward connection kernel, $x_1=0$ and $x_2=0$ and these are same for all $p$. On the left we show positive $p$ values which illustrates that the region containing connected points has a finite area. For negative $p$, we see the area is always infinite, which mimics the size of the future region (forward light cone) in Minkowski space even when bounded by an upper limit on proper-time as this is also infinite. Also note that the curvature of the unit ball, as viewed from the origin, is convex for $p>1$ and concave for $p<1$.
		}
	\label{f:dp-ball}
\end{figure}

\subsection{Volume of Forward Connection Kernel in \Lp{}-cube spacetime}\label{as:kout}

For $p>0$ the boundaries of the forward connection kernel $\Xcal^+(x=0,R)$ of the point at the origin with a length cutoff $R$ are 
\begin{subequations}
	\begin{align}
		x_i&\geq 0
		\label{ae:cubespacerule}
		\\
		R^p&\geq \sum_{i=1}^D (x_i)^p 	
		\label{ae:lengthcutoff}
	\end{align}
\end{subequations}
where the first constraint comes from the cube-space rule and the second from the length cutoff.

Therefore we have that $\Omega_{p>0,D=1}=R$ and
\beq
	\Omega_{p>0,D}(R)
	=
	\int_0^R dx_D \, \Omega_{D-1}( [R^p-(x_{D})^p ]^{1/p} ) 
	\quad \text{if} \;\; D>1
	\, .
\eeq
By dimensions we know that $\Omega_{p>0,D}(R) = R^D\Omega_{p>0,D}(1)$ so we work in terms of the variable $Y$ where
\beq
	Y   = (1-(x_D/R)^p) \, , \quad 
	x_D = R(1-Y)^{1/p}  \, , \quad 
	\frac{dY}{dx_D}  = -\frac{p}{R^p} (x_D)^{p-1} \, .
\eeq
This gives
\bea
	\Omega_{p>0,D}(R)
	&=&
	\int_0^1 dY  \,  \frac{R^p}{p(x_D)^{p-1}}         \Omega_{p>0,D-1}( R Y^{1/p} ) 
	\\
	&=&
	\frac{1}{p}  \Omega_{p>0,D-1}( 1 )  \int_0^1 dY  \,   Y^{(D-1)/p} (1-Y)^{(1/p)-1} \, .
\eea
We recognise the integral as the Beta function
\beq
	B(z,y)
	= \frac{\Gamma(z)\Gamma(y)}{\Gamma(z+y)}
	= \int_0^1 dt \, t^{z-1} (1-t)^{y-1} \, .
	\eeq
	Thus
	\bea
	\Omega_{p>0,D}(1)
	&=&
	\frac{1}{p}  \Omega_{D-1}( 1 )
	\frac{\Gamma(1+ (D-1)/p)\Gamma(1/p)}{\Gamma(1+(D/p))} \, .
\eea
We have that $\Omega_{D=1}(R) = \int_0^R dx_1 = R$, $\Omega_1(1)=1$, so that iterating we have that
\bea
	\Omega_{p>0,D}(R)
	&=&
	R^D \frac{\big( \Gamma(1+(1/p)) \big)^{D}} {\Gamma(1+(D/p))}
	\label{e:Omegadresult}
\eea

For instance, in $D=2$ dimensions we have that
\bea
	\Omega_{p>0,D=2}(R)
	&=&
	R^2 \frac{\big( \Gamma(1+(1/p)) \big)^{2}} {\Gamma(1+(2/p))}
	\label{e:Omegadresultd2}
\eea
which gives (using $\Gamma(z+1) = z \Gamma(z)$, $\Gamma(1) = 1$, $\Gamma(1/2) = \sqrt{\pi}$)
\beq
	\Omega_{p=1,D=2}(R) =  \half R^2 \, , \quad
	\Omega_{p=2,D=2}(R) =  \frac{\pi}{4} R^2 \, , \quad
	\Omega_{p=1/2,D=2}(R) =  \frac{1}{6} R^2
	\label{e:Omegadresultd2pvar}
\eeq
as is well known for the first two cases.

\paragraph{Negative $p$.} For $p<0$, the volume of the forward connection kernel $\Omega_{p<0,D}(R)$ is infnite. However, we are interested in the expected degree in intervals. So a measure of that is to look at the expected out-degree of the source node $s$ of our unit hypercubes $[0,1]^D$:
\begin{equation}
    \kout_s
    =
    \Nbox
    \int_{[0,1]^D}
    d^{D}\xvec \; 
    \Theta \Big( R^p - \sum_{i=1}^{D} z_i^p \Big)
    \, .
    \label{e:expected-out-degree-negative-p}
\end{equation}
Here $\Theta(x)$ the the Heaviside function which is one (zero) for  positive (negative) $x$. 

In our work we are interetested in the case where $p=-1$ and $D=2$. For these values, the boundary of the allowed region satisfies
\begin{equation}
    \frac{1}{x_1}+\frac{1}{x_2}=\frac{1}{R}
    \quad \Leftrightarrow \quad
    x_2=\frac{Rx_1}{x_1-R} \, .
\end{equation}
This curve intersects the upper boundary $x_2=1$ of our unit square $[0,1]^2$ at
$x_1=R/(1-R)$, which lies inside the unit square only for
$R\leq 1/2$. Hence, for $0<R<1/2$, the allowed area gives us
\begin{align}
    \frac{\kout_s}{\Nbox}
    &=
    \int_0^{R/(1-R)} dx_1
    +
    \int_{R/(1-R)}^1 dx_1\, 
    \frac{Rx_1}{x_1-R}\,  \nonumber\\
    &=
    2R
    +
    2R^2
    \ln\!\left(\frac{1-R}{R}\right).
\end{align}
For $R\geq 1/2$, the whole unit square satisfies the cutoff and
$\Omega_{-1,2}(R)=1$. Therefore
\begin{equation}
    \tseexpect{\kout}
    =
    \Nbox
    \begin{cases}
        2R+2R^2\ln\!\left(\dfrac{1-R}{R}\right),
        & 0<R<\dfrac{1}{2},\\[1.2ex]
        1,
        & R\geq\dfrac{1}{2},
    \end{cases}
    \qquad
    p=-1,\quad D=2.
    \label{e:expected-out-degree-pminusone}
\end{equation}